\documentclass[a4paper,11pt]{article}

\usepackage[utf8x]{inputenc}
\usepackage[british]{babel}
\usepackage[raiselinks=false,
            colorlinks=true,
            backref=true,
            linkcolor=violet,
            urlcolor=violet,
            citecolor=orange]{hyperref}

\usepackage{amsmath,qsymbols,mathrsfs}
\usepackage{tabularx}
\usepackage{tikz}
\usepackage{theorem}
\usepackage{url,fullpage}
\usepackage{color}

\definecolor{violet}{RGB}{120,6,250}
\definecolor{orange}{RGB}{250,90,0}

\title{A Categorical Model for\\
  the Lambda Calculus with Constructors}
\author{Barbara Petit}
\date{Focus - INRIA - Università di Bologna\\
  \url{barbara.petit@ens-lyon.org}}

\renewcommand{\abstract}[1]{
  \begin{center}
    \bf Abstract:
  \end{center}
  \hfil
  \parbox{.9\linewidth}{#1}\\
}
\newcommand{\keywords}[1]{
  \vspace{20pt}\noindent
  \textsc{Keywords: }
  \parbox[t]{.84\linewidth}{#1}
}

\newtheorem{theorem}{Theorem}[section]
\newtheorem{lemma}[theorem]{Lemma}
\newtheorem{proposition}[theorem]{Proposition}
\newtheorem{corollary}[theorem]{Corollary}

\newtheorem{definition}{Definition}[section]

\newtheorem{fact}[definition]{Fact}

\newenvironment{proof}{

\noindent
\textsl{Proof:~}}{

}

\newcommand{\qed}{\hfill $\square$}

\newcommand{\ie}{\textit{i.e.}}                  % id est
\newcommand{\cf}{\textit{cf.}}                   % confer
\newcommand{\etc}{\textit{etc.}}                   % etc
\newcommand{\via}{\textit{via}}                   % via
\newcommand{\resp}{\textit{resp.}}               % respectively
\newcommand{\wrt}{\textit{w.r.t.}}               % with respect to
\newcommand{\ssi}{\textit{iff}}              % if and only if
\newcommand{\tosc}{\ensuremath{\rightarrowtriangle}} % red sans ctxt 
\newcommand{\lc}{\ensuremath{`l_{\mathscr{C}}}} % lamC-calc
\newcommand{\eqlc}{\ensuremath{`~_{\lc}}}       % equivalence ds lc
\newcommand{\lcm}{\ensuremath{`l^-_{\mathscr{C}}}} % restr. lamC-calc
\newcommand{\fv}[1]{\ensuremath{\mathrm{fv}(#1)}}        % free term variables
\newcommand{\dom}[1]{\ensuremath{\mathrm{dom}(#1)}}      % domain of CB
\newcommand{\subs}[2][x]{\ensuremath{[#1:=#2]}} % term substitution
\newcommand{\co}[1][c]{\ensuremath{\mathsf{#1}}}  % constructors
\newcommand{\terms}{\ensuremath{`L}}            %set of terms
\newcommand{\rul}[1]{\ensuremath{\text{\sc #1}}}  % regle de red.
\newcommand{\cas}[2][`q]{\ensuremath{
        \{\hspace{-2pt}|#1|\hspace{-2pt}\}\cdot #2 }} % case construct

\newcommand{\C}{\ensuremath{\mathbb{C}}}        % category C
\newcommand{\ccc}{\ensuremath{\textup{CCC}}}         % cat. cart. closed

\newcommand{\uno}{\ensuremath{\mathbf{1}}}                   % objet terminal
\newcommand{\term}[1][]{\ensuremath{!_{#1}}}     % mph vers l'objet terminal
\newcommand{\pdt}{\ensuremath{D^n}}     % produit
\newcommand{\proj}[2][]{\ensuremath{\pi_{#2}^{#1}}} 
\newcommand{\eq}{\ensuremath{\cong}}        % isomorphisme d'objets

\newcommand{\D}{\ensuremath{D}}                   % objet reflexif
\newcommand{\ev}{\ensuremath{\mathtt{ev}}}        % evaluation
\newcommand{\lam}{\ensuremath{\mathtt{lam}}}
\newcommand{\app}{\ensuremath{\mathtt{app}}}
\newcommand{\case}{\ensuremath{\mathtt{case}}}
\newcommand{\fc}[1][c]{\ensuremath{#1^*}}         % fleche associee a un const.
\newcommand{\fail}{\ensuremath{\lightning}}       % match-failure

\newcommand{\id}[1][]{\ensuremath{I\!d}_{#1} }
\newcommand{\comp}{\ensuremath{\bullet}}          % compo. de case-binding
\newcommand{\abstr}[1]{\ensuremath{#1^{`o}}}
\newcommand{\tuple}[1]{\ensuremath{\langle #1 \rangle}} % tuple de mph
\newcommand{\pair}[2][]{\ensuremath{\llparenthesis #2 \rrparenthesis_{#1}}}        % paire (d'elts)
\newcommand{\itp}
[2][`G]{\ensuremath{{[#2]_{#1}}}}   % interp. of terms
\newcommand{\fun}[2][x]{\ensuremath{{\hat{#1}.#2}}}   % notation for functions

\newcommand{\drefRefl}{\ensuremath{(D1)}}
\newcommand{\drefCO}{\ensuremath{(D2)}}
\newcommand{\drefCA}{\ensuremath{(D3)}}
\newcommand{\drefCL}{\ensuremath{(D4)}}
\newcommand{\drefCC}{\ensuremath{(D5)}}
\newcommand{\drefFail}{\ensuremath{(D6)}}

\newcommand{\PER}{\ensuremath{\mathbb{P}\textsc{er}_{\lc}}}      % Cat. des per
\newcommand{\per}{\ensuremath{\lc\mathrm{-}per}}        % part. eq. rel
\newcommand{\M}{\ensuremath{\mathscr{M}}}               % categrical model
\newcommand{\Ms}{\ensuremath{\mathscr{M}_{synt}}}       % syntactic model
\newcommand{\rel}[3][D]{\ensuremath{#2=#3:#1}}          % rel between elts
\newcommand{\class}[2][]{\ensuremath{\overline{#2}^{#1}}}% eq. class

\newcommand{\cpl}[1]{\ensuremath{\widetilde{#1}}}   % case-completion
\newcommand{\cnf}[1]{\ensuremath{\Downarrow #1}}    % case-normal form
\newcommand{\tocc}{\ensuremath{→_{\sc cc}}}   % rec. CaseCase
\newcommand{\mes}[1]{\ensuremath{s(#1)}}             % struct. measure

\usetikzlibrary{positioning,calc,snakes,shadows}   
\tikzset{
  every picture/.style={thick,
                        every node/.style={anchor=mid}},
  vcenter/.style = {baseline={([yshift=-.5ex]current bounding box)}},
  ad/.style      = {->,                     % arete courbee a droite
                    bend right,
                    shorten <=0pt},         % distance au noeud 
  ag/.style      = {->,                     % arete courbee a droite
                    bend left,
                    shorten <=0pt},         % distance au noeud 
  serpent/.style = {snake=snake,            % arete en srpent
                    segment amplitude=.4mm,
                    segment length=2mm,
                    line after snake=1mm},               
  to right/.style= {                         % pour des aretes bien horizontales
    to path={ (\tikztostart.mid east)
              --
              (\tikztotarget.mid west)
              \tikztonodes }
          },
  to left/.style={                         % idem vers la gauche
    to path={ (\tikztostart.mid west)
              --
              (\tikztotarget.mid east)
              \tikztonodes }
          },
  m/.style       = {execute at begin node= \small$,% etiquettes math
                    execute at end node=$%
                  },
  tiny m/.style  = {execute at begin node= $\scriptstyle,% etiquettes math petites
                    execute at end node=$%
},
  etoile/.style  = {at end,            % petite etoile a la fin de la fleche
                    font=$\scriptstyle *$,
                    auto=left},
  arete/.style = {->,shorten >=2pt,thick},
  }

\begin{document}
\maketitle
\abstract{
  The lambda calculus with constructors is an extension of the lambda calculus with variadic constructors. It decomposes the pattern-matching à la ML into a case analysis on constants and a commutation rule between case and application constructs. Although this commutation rule does not match with the usual computing intuitions, it makes the calculus expressive and confluent, with a rather simple syntax.
In this paper we define a sound notion of categorical model for the lambda calculus with constructors. We then prove that this definition is complete for the fragment of the calculus with no match-failure, using the model of partial equivalence relations. 
}


\section*{Introduction}
\label{sec:intro}

Pattern matching is now a key feature in most functional programming languages.
Inherited from the simple constants recognition mechanism that appeared in the late 60's  (in \textsl{Snobol} or in \textsl{Pascal} for instance), it is now a elaborated feature in main programming languages (\textsl{ML, Haskell} \etc) and some proof assistants (such as \textsl{Coq} or \textsl{Agda}), able to decompose complex data-structures.

Its theoretical aspects are being intensively studied since the 90's~\cite{Coquand92,OngRamsey11}.
In particular, several lambda calculi with pattern matching have been proposed~\cite{Oostrom90,CirKir98,JayKes06}.
Among them, the lambda calculus with constructors~\cite{AAA06} (or \lc-calculus) offers the advantage of having simple computation rules.
Indeed, the pattern matching \textit{à la} ML is there decomposed into two atomic rules (a constants analysis rule, and a commutation rule).
The rather simple syntax of this calculus together with the decomposition of its powerful computational behaviour into elementary steps stimulate a semantic study of the the \lc-calculus from a categorical point of view.

As far as we know, no categorical model had been proposed so far for a calculus with pattern matching.
Yet category theory allows to express some generic semantic properties on a calculus, and to factorise many of its different concrete models.
Furthermore, when the categorical model is \emph{complete}, it synthesises exactly the extensional properties of the calculus.
Since the description of the models for the pure lambda calculus as Cartesian closed categories with a reflexive object~\cite{Scott80}, 
some complete categorical models have been defined for variants of the lambda calculus~\cite{HofStreich97,Selinger01,FauMiq09}.

In this paper, after a brief presentation of the \lc-calculus (Sec.~\ref{sec:intro}), we establish a categorical definition of models for it (Sec.~\ref{sec:cat}).
We then prove that it is to some extent complete for the \lc-calculus, using the standard PER model and some rewriting techniques (Sec.~\ref{sec:per}).
Notice that we only use very basic notions of category theory (knowledge of the first two chapters of~\cite{AsperLongo91} is sufficient).

\section{The lambda calculus with constructors}
\label{sec:lc}

The lambda calculus with constructors extends the pure lambda calculus with pattern matching features:
a set of constants (that we consider here to be finite of cardinal~$n$) called \emph{constructors} and denoted by \co, \co[d] \etc\ is added, with a simple mechanism of case analysis on these constants (similar to the \texttt{case} instruction of Pascal):
%\\
\vspace{5pt}

\hfill
\begin{math}
  \cas[\co_1\mapsto t_1;\dots;\co_k\mapsto t_k]{\co_i}
  \quad → \quad t_i \qquad
\end{math}
\hfill
(\rul{CaseCons})
%\\

\noindent
Although only constant constructors can be analysed, a matching on variant constructors can be performed \via\ a commutation rule between case construction and application:
%\\

\hfill
\begin{math}
  \cas{(tu)} \quad→\quad (\cas{t})~u
\end{math}
\hfill (\rul{CaseApp})
%\\

This commutation rule enables simulating any pattern matching \textit{à la} ML, by generalising the following example: 
in the \lc-calculus, the predecessor function on unary integers (represented with the constructors \co[0] and \co[S]) is implemented as
\begin{math}
  \mathtt{pred}=`lx.
  \cas[{\co[0]\mapsto \co[0];\co[S]\mapsto `ly.y}]{x}.
\end{math}\quad
Applying this function to a non zero integer~$\co[S]\,n$ actually produces the expected result:
\begin{displaymath}
  \begin{array}[t]{ccl}
    \mathtt{pred}~(\co[S]~m) &→& 
    \cas[{\co[0]\mapsto \co[0];\co[S]\mapsto `ly.y}]{(\co[S]~m)}\\
    &→&(\cas[{\co[0]\mapsto \co[0];\co[S]\mapsto `ly.y}]{\co[S]})~m
    ~→~ (`ly.y)~m ~→~ m
  \end{array}
\end{displaymath}

Formally, the syntax of the \lc-calculus is defined by the following grammar:
\\

\hfill
\begin{math}
  \begin{array}[b]{c@{\ :=\quad}l}
    t, u, v & x ~~|~~tu ~~|~~`lx.t~~|~~\co~~|~~\cas{t} \\
    `q,`f & \{\co_1 \mapsto u_1;\hdots; \co_k \mapsto u_k\} 
  \end{array}
  \scriptstyle
  (\text{with } k\geq 0 \text{ and } \co_i \neq \co_j 
  \text{ for } i \neq j)
\end{math}
\\

\noindent
In the terms (denoted by~$t,u$ \etc) the application takes precedence over lambda abstraction and case construct.
Notice that constructors, like any terms, can be applied to any number of arguments and thereby are \emph{variadic} (they have no fix arity).
We call \emph{data-structure} a term on the form $\co\,t_1\cdots t_k$.

A \emph{case-binding}~$`q$ is just a (partial) function from constructors to terms, whose domain is written~\dom{`q}.
By analogy with sequential notation, we may write~$`q_c$ for~$u$ when $\co\mapsto u `: `q$.
In order to ease the reading, we may write
$\cas[\co_1\mapsto u_1;\dots;\co_n\mapsto u_n]{t}$ instead of 
$\cas[\{\co_1\mapsto u_1;\dots;\co_n\mapsto u_n\}]{t}$.
The usual definition of the free variables of a term is naturally extended to the new constructions of the calculus, taking care that constructors are not variables (and therefore not subject to substitution nor $`a$-conversion).

In this calculus, a \emph{match failure} is a term~\cas{\co}\ where $\co`;\dom{`q}$.
We say that a term is \emph{defined} when none of its subterm is a match failure, and that it is \emph{hereditarily defined} when all this reducts (in any number of steps, including zero) are defined.

Reduction  rules are given in Fig.~\ref{fig:rules}.
In addition to the usual $`b$-reduction (called \rul{AppLam}) and to the two rules presented earlier, there is a rule of commutation between case construct and lambda abstraction (\rul{CaseLam}) to ensure confluence~\cite[Cor.~1]{AAA06}, and the usual $`h$-reduction (called \rul{LamApp}) as well as a rule of composition of case-bindings (\rul{CaseCase}) so that the calculus enjoys the \emph{separation property}~\cite[Theo.~2]{AAA06}.
More explanations and examples about this calculus can be found in~\cite{AAA09,PetitTez}.

\begin{figure}[hbt]
    \hrule\vspace{1pt}\hrule
  \centering
%  \framebox[\textwidth]{
    \begin{math}
      \begin{array}
        {@{\qquad}l@{\quad}l@{\quad\quad}rcl@{\qquad}l@{\qquad}}
        &&&&\\
%        \multicolumn{6}{l}{\textbf{Beta-Eta reduction}}\\
        \rul{ AppLam} & (\rul{AL}) &
        (`lx.t)u    &\tosc& t\subs{x}{u} & \\
        \rul{ LamApp} & (\rul{LA}) &
        `lx.tx &\tosc& t & (x\notin \fv{t}) \\
        &&&&&\\
%        \multicolumn{6}{l}{\textbf{Case propagation}}\\
        \rul{ CaseCons} & (\rul{CO}) &
        \cas{\co} &\tosc& t & ((\co\mapsto t)\in`q) \\
        \rul{ CaseApp} & (\rul{CA}) &
        \cas{(tu)} &\tosc& (\cas{t})u & \\
        \rul{ CaseLam} & (\rul{CL}) &
        \cas{`lx.t} &\tosc&
        `lx.\cas{t} & (x\notin \fv{`q}) \\
%        &&&&&\\
%        \multicolumn{6}{l}{\textbf{Case composition}}\\
        \rul{ CaseCase} & (\rul{CC}) &
        \cas{ \cas[`f]{t} } &\tosc&
        \cas[`q\circ`f]{t} &\\
        \multicolumn{6}{r}{\qquad \text{with }
          `q \circ \big\{\co_1 \mapsto t_1;...; \co_n \mapsto t_n\big\}
          =\big\{ \co_1\mapsto \cas{t_1};...;\co_n \mapsto \cas{t_n}\big\}
          \qquad
        }\\
        &&&&&\\
      \end{array}
    \end{math}
    \hrule\vspace{1pt}\hrule
%  }
  \caption{Reduction rules for $\lc$.}
  \label{fig:rules}
\end{figure}

\section{The categorical model}
\label{sec:cat}
In this section we may define a notion of a categorical model for the \lc-calculus, that we prove to be sound.
No deep knowledge in category theory is assumed from the reader, he might just know the definition of a Cartesian closed category (also said a \ccc).

The notations we use are quite standard:
in a \ccc, the product of two objects~$A$ and~$B$ is written~$A`*B$ and their exponential~$B^A$.
The $k$-ary product of~$A$ is denoted by~$A^k$, and the identity morphism on~$A$ by~$\id[A]$ (or simply~$\id$ if it raises no ambiguity).
The $i^{th}$ projection morphism of a $k$-ary product is written~\proj[k]{i}, or~\proj{i} if $k=2$.
Given some morphisms $f:A→B$, $g:A→C$ and $h:A→C$, \tuple{f;g} denotes the pairing of~$f$ and~$g$, and~$f;h$ the composition of~$f$ and~$h$.
The evaluation map of~$A$ and~$B$ is $\ev:B^A`*A→B$ and the curried form of a morphism~$f$ is written~$`L(f)$.

\subsection{\lc-models}

It is well known~\cite{MitchellBook} that Cartesian closed categories have exactly the good structure to interpret the typed lambda calculus.
To cope with the problem of self application of terms, such a category must be provided with a reflexive object~\D\ in order to interpret the \textit{untyped} lambda calculus~\cite{Scott80}.
Terms are then interpreted by points of~\D.
The denotation of applications is constructed with a morphism $\app:\D→\D^\D$, and the one of lambda abstractions with a morphism $\lam:\D^\D→\D$.
Also the correction of the $`b$-reduction is ensured by the equality $\lam;\app=\id[\D^\D]$
(if moreover $\app;\lam=\id[\D]$, then the model satisfies the $`h$-equivalence).

Building a model for the \lc-calculus requires some extra morphisms and equalities for the new constructions and the new rules of the calculus.
In particular, writing $\{\co_1,\dots\co_n\}$ the set of constructors, a special point~$\fc_i$ of~\D\ is needed for each $i\leq n$ to interpret them.
The denotations of case-bindings are then points of~\pdt.
\label{pg:itp-cb}
A case binding~$`q$ is interpreted by the $n$-tuple $\tuple{d_1;\dots;d_n}$ where~$d_i$ is the denotation of~$`q_{\co_i}$ if $\co_i`:\dom{`q}$, and is a special point~\fail\ representing match failure otherwise.
In order to interpret case constructs, we need a morphism $\case:\pdt`*\D→\D$, that transforms the denotation of~$`q$ and~$t$ into the one of $\cas{t}$.

Let us informally confuse terms and their denotations, and write a case-binding~$\{\co_i\mapsto u_i/1\leq i\leq n\}$ as~$\{\vec{\co}\mapsto\vec{u}\}$ and its denotation as~$\vec{u}$.
Then the rule \rul{CaseCons} is valid if $\cas[\vec{\co}\mapsto\vec{u}]{\co_i}$ and~$u_i$ have the same denotation, \ie\ intuitively if $\case(\vec{u},\co_i)=\proj[n]{i}(\vec{u})$.
This is formally expressed by the commutation of the diagram~\drefCO\ in Fig.~\ref{fig:com-diag}.

\noindent
In the same way, the rule \rul{CaseApp} is valid if the diagram~\drefCA\ commutes, \ie~if

\hfil
%\begin{center}
\begin{tikzpicture}[baseline=(a.base)]
  \matrix[column sep=30pt]{
    \node[m] (a) {\pdt`*D`*D}; &
    \node[m] (b) {\pdt`*D^D`*D}; &
    \node[m] (c) {\pdt`*D}; &
    \node[m] (d) {D}; \\
    \node[m] {(\vec{u},t,t')}; &
    \node[m] {(\vec{u}\,,\,\fun{tx}\,,\,t')}; &
    \node[m] {(\vec{u}~,~tt')}; &
    \node[m] {\cas[\vec\co\mapsto\vec{u}]{(tt')}}; \\
  };
  \draw[->,to right] (a) to node[above,tiny m] {\_`*\app`*\_} (b);
  \draw[->,to right] (b) to node[above,tiny m] {\_`*\ev} (c);
  \draw[->,to right] (c) to node[above,tiny m] {\case} (d);
\end{tikzpicture}
%\end{center}
\\
(where $\fun{v}$ represents the function mapping $v_0$ to $v\subs{v_0}$) is equal to

\hfil
%\begin{center}
\begin{tikzpicture}[baseline=(a.base)]
  \matrix[column sep=15pt]{
    \node[m] (a) {\pdt`*D`*D}; &
    \node[m] (b) {D`*D}; &
    \node[m] (c) {D^D`*D}; &
    \node[m] (d) {D}; \\
    \node[m] {(\vec{u},t,t')}; &
    \node[m] {(\cas[\vec\co\mapsto\vec{u}]{t}~,~t')}; &
    \node[m] {\big(\fun{(\cas[\vec\co\mapsto\vec{u}]{t})x}~,~t'\big)}; &
    \node[m] {(\cas[\vec\co\mapsto\vec{u}]{t})~t'}; \\
  };
  \draw[->,to right] (a) to node[above,tiny m] {\case`*\_} (b);
  \draw[->,to right] (b) to node[above,tiny m] {\app`*\_} (c);
  \draw[->,to right] (c) to node[above,tiny m] {\ev} (d);
\end{tikzpicture}  
%\end{center}

To express the rule \rul{CaseLam} we need a morphism that abstracts the case construct \wrt\ a variable:

\hfil
\begin{math}
  \abstr\case=`L(f_{\case}):
  \begin{array}[t]{ccl}
    \pdt`*D^{D} & → & D^{D} \\
    (\vec{u},\fun{t}) & \mapsto &
    \fun{~\cas[\vec\co\mapsto\vec{u}]{t}}
  \end{array}
\end{math}
\\
where $f_{\case}=$
\begin{tikzpicture}[vcenter]
  \matrix[column sep=10pt]{
    \node[m] (a) {(\pdt`*D^{D})`*D}; &
    \node[m] (b) {\pdt`*(D^{D}`*D)}; &&&
    \node[m] (c) {\pdt`*D}; &&
    \node[m] (d) {D}; \\
  };
  \draw[->] (a) to node[above,tiny m] {\eq} (b);
  \draw[->] (b) to node[above,tiny m] {\id[\pdt]`*\ev} (c);
  \draw[->] (c) to node[above,tiny m] {\case} (d);
\end{tikzpicture}.
\\
Then the rule \rul{CaseLam} is valid if \drefCL\ commutes:
\\
\begin{tikzpicture}[baseline=(a.base)]
  \matrix[column sep=5pt]{
    \node[m] (a) {\pdt`*D^{D}}; &
    \node[m] (b) {D^{D}}; &
    \node[m] (c) {D}; \\
    \node[m] {(\vec{u},\fun{t})}; &
    \node[m] {\fun{\cas[\vec\co\mapsto\vec{u}]{t}}}; &
    \node[m] {`lx.\cas[\vec\co\mapsto\vec{u}]{t}}; \\
  };
  \draw[->,to right] (a) to node[above,tiny m] {\abstr\case} (b);
  \draw[->,to right] (b) to node[above,tiny m] {\lam} (c);
\end{tikzpicture}
\hspace{-15pt}=
\begin{tikzpicture}[baseline=(a.base)]
  \matrix[column sep=5pt]{
    \node[m] (a) {\pdt`*D^{D}}; &&&
    \node[m] (b) {\pdt`*D}; &
    \node[m] (c) {D}; \\
    \node[m] {(\vec{u},\fun{t})}; &&&
    \node[m] {(\vec{u},`lx.t)}; &
    \node[m] {\cas[\vec\co\mapsto\vec{u}]{`lx.t}}; \\
  };
  \draw[->,to right] (a) to node[above,tiny m] {\_`*\lam} (b);
  \draw[->,to right] (b) to node[above,tiny m] {\case} (c);
\end{tikzpicture}

\vspace{10pt}
\noindent
Also the rule \rul{CaseCase} requires a morphism to compose case-bindings:
\begin{displaymath}
  \comp:
  \begin{array}[t]{ccl}
    \pdt`*\pdt & → & \pdt \\
    (\vec{u},(t_i)_{i=1}^n) & \mapsto &
    (\cas[\vec\co\mapsto\vec{u}]{t_i})_{i=1}^n \\
  \end{array}
\end{displaymath}
It is defined as the pairing of the morphisms $(\id[\pdt]`*\proj[n]{i});\case$, for $1\leq i\leq n$.
So it is the unique morphism that makes the diagram on the following commute.
\begin{center}
  \begin{tikzpicture}
    \matrix[row sep=10pt,column sep=10pt]{
      & \node (pp) {$\pdt`*\pdt$}; & \\
      \node (a1) {$\pdt`*D$}; & \node (a) {$\cdots$}; & \node (an) {$\pdt`*D$}; \\
      \node (b1) {D}; & \node (b) {$\cdots$}; & \node (bn) {D}; \\
      & \node (p) {$\pdt$}; & \\
    };
    \draw[->] (pp) to
    node [above left] {$\scriptstyle\id[]`*\proj[n]{1}$} (a1);
    \draw[->] (pp) to
    node [above right] {$\scriptstyle\id[]`*\proj[n]{n}$} (an);
    \draw[->] (a1) to 
    node [left] {$\scriptstyle\case$} (b1);
    \draw[->] (an) to 
    node [right] {$\scriptstyle\case$} (bn);
    \draw[->] (p) to 
    node[below left] {$\scriptstyle \proj[n]{1}$} (b1);
    \draw[->] (p) to 
    node[below right] {$\scriptstyle \proj[n]{n}$} (bn);
    \draw[ad,dashed] (pp) to 
    node [left] {$\comp$} (p);
  \end{tikzpicture}
\end{center}

\noindent
Then the commutation of the diagram~\drefCC\ validates the rule \rul{CaseCase}.

This leads to the following definition.
\begin{definition}[\lc-model]
  A categorical model for the untyped~\lc-calculus is\\
  \begin{math}
    \M = (\C\,,D\,,\app\,,\lam\,,(\fc[c_i])^n_{i=1},\fail,\case)
  \end{math} where
  \begin{itemize}
  \item \C\ is a Cartesian closed category,
  \item \D\ is an object of~\C,
  \item All the $\fc_i$'s and~\fail\ are points of~$D$,
  \item\app\ is a morphism of~$D→D^D$, \lam\ is a morphism of~$D^D→D$ and \case~a morphism of~$\pdt`*D→D$,
  \item The six diagrams of Fig.~\ref{fig:com-diag} commute (the diagram~\drefCO\
    must commute for every $i`:\llbracket 1..n\rrbracket$). 
  \end{itemize}
\end{definition}

\begin{figure}[ht]
  \centering
  \begin{tabular*}{1.0\linewidth}[t]
    {|@{\,}|m{20pt}|m{205pt}|@{\,}|m{20pt}|m{152pt}|@{\,}|}
    \hline
    \multicolumn{2}{|@{\,}|c|@{\,}|}{\rul{LamApp/AppLam}}&
    \multicolumn{2}{c|@{\,}|}{\rul{CaseCons}}\\
    \hline
    $\!\!$\drefRefl &\qquad
    \begin{tikzpicture}[vcenter]
      \node (d) {$D$};
      \node (dd) [right =of d] {$D^{D}$};
      % 
      % \path (d) to node {$\eq$} (dd);
      \draw [->,bend left] (d) to 
      node [above] (a) {$\scriptstyle\app$} (dd);
      \draw [->,bend left] (dd) to 
      node [below] (l) {$\scriptstyle\lam$} (d);
      \draw [->] ([yshift=2.5ex]d.base) arc (45:315:1.5em and 1em);
      % (pt depart) arc (angle debut:angle fin:rayonH and rayonV)
      \node [left= of a] {$\scriptstyle\id[D]$};
      \draw [->] ([yshift=2.5ex]dd.base) arc (135:-135:1.5em and 1em);
      \node [right=of a] {$\scriptstyle\id[D^D]$};
    \end{tikzpicture}
    &
    $\!\!$\drefCO &
    \quad
    \begin{tikzpicture}[vcenter]
      \node (prod) {$\pdt$};
      \node (prod1) [right = of prod]{$\pdt`*\uno$};
      \node (d) [below = of prod]{$D$};
      \node (dprod) [below = of prod1]{$\pdt`*D$};
      \draw [<->] (prod) to
      node [label=above:{$\scriptstyle\eq$}] {} (prod1);
      \draw [->] (prod) to
      node [label=left:{$\scriptstyle\proj[n]{i}$}] {} (d);
      \draw [->] (prod1) to
      node [label=right:{$\scriptstyle\id[]`*\fc[c_i]$}] {} (dprod);
      \draw [->] (dprod) to
      node [label=below:{$\scriptstyle\case$}] {} (d);
    \end{tikzpicture}
    \\
    \hline
    \hline
    \multicolumn{2}{|@{\,}|c|@{\,}|}{\rul{CaseApp}}&
    \multicolumn{2}{c|@{\,}|}{\rul{CaseLam}}\\
    \hline
    $\!\!$\drefCA &
    \begin{tikzpicture}[vcenter]
      % every node/.style={anchor=mid},]
      \node (g)  {$(\pdt`*D)`*D$};
      \node (g1) [below= of g] {$D`*D$};
      \node (g2) [below= of g1] {$D^D`*D$};
      \node (d)  [right=.5 of g] {$\pdt`*(D`*D)$};
      \node (d1) at (g1 -| d) {$\pdt`*(D^{D}`*D)$};
      \node (d2) at (g2 -| d) {$\pdt`*D$};
      \node (c)  [below= of $(g2)!0.5!(d2)$] {$D$};
      \draw [->] (g) to
      node [label=left:{$\scriptstyle\case`*\id[]$}] {} (g1);
      \draw [->] (g1) to
      node [label=left:{$\scriptstyle\app`*\id[]$}] {} (g2);
      \draw [->] (g2) to
      node [label=below left:{$\scriptstyle\ev$}] {} (c);
      \draw [<->] (g) to
      node [label=above:{$\scriptstyle\eq$}] {} (d);
      \draw [->] (d) to
      node [label=left:{$\scriptstyle\id[]`*(\app`*\id[])$}] {} (d1);
      \draw [->] (d1) to
      node [label=left:{$\scriptstyle\id[]`*\ev$}] {} (d2);
      \draw [->] (d2) to
      node [label=below right:{$\scriptstyle\case$}] {} (c);
    \end{tikzpicture}&
    $\!\!$\drefCL &
    \quad
    \begin{tikzpicture}[vcenter]
      \node (pdd) {$\pdt`*D^{D}$};
      \node (dd)  [right= of pdd] {$D^{D}$};
      \node (pd)  [below= of pdd] {$\pdt`*D$};
      \node (d) at (pd -| dd)  {$D$};
      \draw [->] (pdd) to 
      node[above] {$\scriptstyle\abstr{\case}$} (dd);
      \draw [->] (pdd) to 
      node[left] {$\scriptstyle\id[]`*\lam$} (pd);
      \draw [->] (dd) to 
      node[right] {$\scriptstyle \lam$} (d);
      \draw [->] (pd) to 
      node[below] {$\scriptstyle \case$} (d);
    \end{tikzpicture}
    \\
    \hline
    \hline
    \multicolumn{4}{|@{\,}|c|@{\,}|}{\rul{CaseCase}}\\
    \hline
    $\!\!$\drefCC &
    \begin{tikzpicture}[vcenter]
      \node (g)  {$(\pdt`*\pdt)`*D$};
      \node (g1) [below= of g] {$\pdt`*D$};
      \node (d)  [right= 0.5 of g] {$\pdt`*(\pdt`*D)$};
      \node (d1) at (g1 -| d) {$\pdt`*D$};
      \node (c)  [below= of $(g1)!0.5!(d1)$] {$D$};
      \draw [->] (g) to
      node [label=left:{$\scriptstyle\comp`*\id[]$}] {} (g1);
      \draw [->] (g1) to
      node [label=below left:{$\scriptstyle\case$}] {} (c);
      \draw [<->] (g) to
      node [label=above:{$\scriptstyle\eq$}] {} (d);
      \draw [->] (d) to
      node [label=left:{$\scriptstyle\id[]`*\case$}] {} (d1);
      \draw [->] (d1) to
      node [label=below right:{$\scriptstyle\case$}] {} (c);
    \end{tikzpicture}
    &
    $\!\!$\drefFail&
    \begin{tikzpicture}[vcenter]
      \matrix[column sep=3em,row sep=2em]
      { \node (a) {$\pdt`*\uno$}; &
        \node (b) {$\pdt`*D$}; \\
        \node (c) {$\uno$}; &
        \node (d) {$D$}; \\        
      };
      \draw[->] (a) to node [above] {$\scriptstyle\id[\pdt]`*\fail$}(b);
      \draw[->] (c) to node [below] {$\scriptstyle\fail$} (d);
      \draw[->] (a) to node [left] {$\scriptstyle\proj{2}$} (c);
      \draw[->] (b) to node [right] {$\scriptstyle\case$} (d);
    \end{tikzpicture}
    \\
    \hline
  \end{tabular*}
  \caption{Commuting diagrams in a \lc-model}
  \label{fig:com-diag}
\end{figure}

\paragraph{Equivalent definition.}

In fact we can simplify the definition of a \lc-model, since the isomorphism~$\D\eq\D^\D$ entails the equivalence of the diagrams~\drefCA\ and~\drefCL.
This can be understood from a syntactical point of view, given that
the commutation of the diagram~\drefCA\ validates the rule \rul{CaseApp} and the one of~\drefCL\ validates \rul{CaseLam}.
Indeed, the only role of \rul{CaseLam} in the calculus is to close a critical pair created by the rule \rul{CaseApp}~\cite[Theo.~1,~\textit{(CC3)}]{AAA06}.

\begin{proposition}
  \label{prop:equiv-cacl}
  If~\lam\ and~\app\ form an isomorphism between~$D$ and~$D^{D}$, then the diagram~\drefCA\ commutes if and only if the diagram~\drefCL\ commutes.
\end{proposition}
\begin{proof}~
  \\
  \parbox[t]{0.6\linewidth}{
    Since~\drefRefl\ commutes, \drefCL\ commutes \ssi\ the diagram on the right commutes.
    \\
    Write $f=\id[\pdt]`*\lam;\case;\app$.
    \\
    Since~$\abstr{\case}=`L(\eq;\id[\pdt]`*\ev;\case)$, and by     uniqueness of the exponential, $f=\abstr{\case}$ if and only if the following diagram commutes:
  }
  \hfill
  \parbox[t]{0.3\linewidth}{
    \begin{tikzpicture}[baseline=(pdd.north)]
      \node (pdd) {$\pdt`*D^{D}$};
      \node (dd)  [right= of pdd] {$D^{D}$};
      \node (pd)  [below= of pdd] {$\pdt`*D$};
      \node (d) at (pd -| dd)  {$D$};
      \draw [->] (pdd) to 
      node[above] {$\scriptstyle\abstr{\case}$} (dd);
      \draw [->] (pdd) to 
      node[left] {$\scriptstyle\id[]`*\lam$} (pd);
      \draw [->] (d) to 
      node[right] {$\scriptstyle \app$} (dd);
      \draw [->] (pd) to 
      node[below] {$\scriptstyle \case$} (d);
    \end{tikzpicture}
  }
  
  \hfil
  \begin{tikzpicture}[node distance=10pt]
    \node (a) {$(\pdt`*D^D)`*D$};
    \node [right=2.5 of a] (b) {D};
    \node [below=of a] (c) {$D^D`*D$};
    \draw [->] (a) to 
    node[above] {$\scriptstyle \eq~;~\id[\pdt]`*\ev~;~\case$} (b);
    \draw [->] (a) to node[left] {$\scriptstyle f`*\id[D]$} (c);
    \draw [->] (c) to node[below right] {$\scriptstyle \ev$} (b);
  \end{tikzpicture}
  
  %\pagebreak[4]
  \noindent
  We can detail this diagram as follows:
  \begin{center}
    \begin{tikzpicture}[node distance=15pt]
      \node (a1) {$(\pdt`*D^D)`*D$};
      \node [right=of a1] (a2) {$\pdt`*(D^D`*D)$};
      \node [right=of a2] (a3) {$\pdt`*D$};
      \node [below=of a1] (b1) {$(\pdt`*D)`*D$};
      \node (b2) at (b1 -| a2) {$\pdt`*(D`*D)$};
      \node [below=of b1] (c1) {$D`*D$};
      \node (c2) at (c1 -| a2) {$D^D`*D$};
      \node (c3) at (c1 -| a3) {$D$};
      \draw [ad] (a1) to 
      node [left] {$\scriptstyle (\id[]`*\lam)`*\id$} (b1);
      \draw [ad] (b1) to 
      node [left,label=right:{$\scriptstyle         (\id[]`*\app)`*\id$}]{$\eq\,$}(a1);
      \draw [->] (a1) to node[above,to right] {$\scriptstyle \eq$} (a2);
      \draw [->] (a2) to node[above] {$\scriptstyle \id[\pdt]`*\ev$} (a3);
      \draw [->] (a3) to node[right] {$\scriptstyle \case$} (c3);
      \draw [->] (b1) to node[left] {$\scriptstyle \case`*\id[D]$}       (c1);
      \draw [->] (c1) to node[below,to right] {$\scriptstyle \app`*\id[D]$} (c2);
      \draw [->] (c2) to node[below,to right] {$\scriptstyle \ev$} (c3);
      \draw [->] (b1) to node[below,to right] {$\scriptstyle \eq$} (b2);
      \draw [->] (b2) to 
      node[right] {$\scriptstyle \id[\pdt]`*(\app`*\id[D])$} (a2);
      \node at ($(b2)!0.5!(c3)$) {\drefCA};
      \node at ($(b1)!0.4!(a3)$) {\LARGE$\circlearrowright$};
    \end{tikzpicture}
  \end{center}
  Since the sub-diagram in the upper-left corner commutes, then \drefCL\ commutes if and only if \drefCA\ commutes.
  \qed
\end{proof}
Thus we can omit the commutation of~\drefCA\ or the one of~\drefCL\ in the definition of a \lc-model.

\subsection{Soundness}
\label{sec:sound}

In the previous section we gave some intuitions on how to interpret \lc-terms in a \lc-model.
Formally, the denotation~\itp{t} of a term~$t$ in such a category is defined by structural induction (in Fig.~\ref{fig:itp-cat}).
It depends on a list of variables~$`G=x_1,\cdots,x_k$ that must contain all the free variables of~$t$, and its a morphism of $\D^k→\D$.
Similarly, the denotation~\itp{`q} of a case-binding~$`q$ with free variables in~$`G$ is a morphism of $\D^k→\pdt$.
We show that this definition provides a correct model of the \lc-calculus (we write~\eqlc\ for the reflexive symmetric transitive closure of its six rules).

\begin{figure}[htb]
  \centering
  \hrule\vspace{1pt}\hrule\vspace{10pt}
  \begin{tabular}[t]{c@{\ =\ }l}
    \itp{x_i}&$\proj[k]{i}:D^k→D$\\
    \itp{tu}&
    \begin{tikzpicture}
      [baseline={([yshift=-1ex]current bounding box)}]
      \matrix[anchor=south,column sep=4em]
      {
        \node (a) {$D^k$}; &
        \node (b) {$D`*D$}; &
        \node (c) {$D^D`*D$}; &
        \node (d) {$D$};\\
      };
      \draw [->,to right] (a) to 
      node[above] {$\scriptstyle\tuple{\itp{t};\itp{u}}$} (b);
      \draw [->,to right] (b) to 
      node[above] {$\scriptstyle \app`*\id[D]$} (c);
      \draw [->,to right] (c) to 
      node[above] {$\scriptstyle \ev$} (d);
    \end{tikzpicture}\\
    \itp{`lx_{k+1}.t}&
    \begin{tikzpicture}
      [baseline={([yshift=-1ex]current bounding box)}]
      \matrix[column sep=4em]
      {
        \node (a) {$D^k$}; &
        \node (b) {$D^D$}; &
        \node (c) {$D$}; \\
      };
      \draw [->,to right] (a) to 
      node[above] {$\scriptstyle `L(f_t)$} (b);
      \draw [->,to right] (b) to 
      node[above,to right] {$\scriptstyle \lam$} (c);
    \end{tikzpicture} \\
    where $f_t$ &
    \begin{tikzpicture}
      [baseline={([yshift=-1ex]current bounding box)}]
      \matrix[column sep=3em]
      {
        \node (a) {$D^k`*D$}; &
        \node (b) {$D^{k+1}$}; &
        \node (c) {$D$};\\
      };
      \draw [->,to right] (a) to 
      node[above] {\eq} (b);
      \draw [->,to right] (b) to 
      node[above] {$\scriptstyle \itp[`G,x_{k+1}]{t}$} (c);
    \end{tikzpicture}\\
    \itp{\co[c]} &
    \begin{tikzpicture}
      [baseline={([yshift=-1ex]current bounding box)}]
      \matrix[column sep=3em]
      {
        \node (a) {$D^k$}; &
        \node (b) {$\uno$}; &
        \node (c) {$D$};\\
      };
      \draw [->,to right] (a) to 
      node[above] {$\scriptstyle \term[D^k]$} (b);
      \draw [->,to right] (b) to 
      node[above] {$\scriptstyle \fc$} (c);
    \end{tikzpicture}\\
    \itp{\cas{t}} &
    \begin{tikzpicture}
      [baseline={([yshift=-1ex]current bounding box)}]
      \matrix[column sep=4em]
      {
        \node (a) {$D^k$}; &
        \node (b) {$\pdt`*D$}; &
        \node (c) {$D$};\\
      };
      \draw [->,to right] (a) to 
      node[above] {$\scriptstyle \tuple{\itp{`q};\itp{t}}$} (b);
      \draw [->,to right] (b) to 
      node[above] {$\scriptstyle \case$} (c);
    \end{tikzpicture}\\
    \itp{`q} &
    $\tuple{f_1;\cdots;f_n}:D^k→\pdt$
    ,\quad
    where $f_i$ = 
    \begin{math}
      \left\{
        \begin{array}[c]{ll}
          \itp{u_i}&\text{if }\co[c_i]\mapsto u_i `:`q\\
          \term[D^k];\fail &\text{if }\co[c_i]`;\dom{`q}
        \end{array}
      \right.
    \end{math}\\
  \end{tabular}
  \vspace{10pt}\hrule\vspace{1pt}\hrule
  \caption{Interpretation of \lc-terms in a categorical model}
  \label{fig:itp-cat}
\end{figure}

\begin{theorem}[Soundness]
  \label{theo:sound}
  If $\M=(\C,D,\lam,\app,(\fc[\co_i])^n_{i=1},\case,\fail)$ is
a \lc-model, then for any \lc-term~$t,t'$ whose free variables are in~$`G$,
  \begin{displaymath}
    t\eqlc t' \quad\implies\quad
    \itp{t}=\itp{t'}
  \end{displaymath}
\end{theorem}

\noindent
To prove this theorem, we fix a \lc-model $\M=(\C,D,\lam,\app,(\fc[\co_i])^n_{i=1},\case,\fail)$ and use some preliminary lemmas.
The first one expresses that the morphism~\comp\ actually corresponds to case-composition.
This is where we technically need the diagram~\drefFail, even though its semantic meaning is not as intuitive as for the other one.

\begin{lemma}[Categorical case-composition]
  \label{lem:cat-cc}
  If the diagram~\drefFail\ commutes, then for any case-bindings~$`q$     and~$`f$, whose free variables are in~$`G=\{x_1,\dots,x_k\}$, the   following diagram commute:
  \begin{center}
    % \begin{equation}
    %   \label{eq:cat-cc}
    \begin{tikzpicture}[vcenter]
      \matrix[column sep=4em,row sep=2em]
      { \node (a) {$D^k$}; &
        \node (b) {$\pdt`*\pdt$}; \\&
        \node (c) {$\pdt$}; \\
      };
      \draw[->](a) to node[above] {$\scriptstyle
        \tuple{\itp{`q},\itp{`f}}$} (b);
      \draw[->](b) to node[right] {$\scriptstyle\comp$} (c);
      \draw[->](a) to node[below left]
      {$\scriptstyle\itp{`q`o`f}$} (c);
    \end{tikzpicture}
    % \end{equation}
  \end{center}
\end{lemma}
\begin{proof}
  If $`f=\{\co[c_i]\mapsto u_i/ i`:J\}$ (with~$J``(=\llbracket   1..n\rrbracket$), then

  \hfill
  \begin{math}
    \itp{`q`o`f}=\tuple{f_1,\dots,f_n}~,\quad
    \text{with } f_i=
    \left\{
      \begin{array}[c]{ll}
        \itp{\cas{u_i}}&\text{if }i`:J\\
        \term[D^k];\fail&\text{if }i`;J
      \end{array}
    \right.
  \end{math}\\
  On the other hand,   $\comp=\tuple{\big((\id[\pdt]`*\proj[n]{1});\case\big),\dots,
    \big((\id[\pdt]`*\proj[n]{1});\case\big)}$.
  So 
  \begin{displaymath}
    \tuple{\itp{`q},\itp{`f}}~;~\comp=\tuple{g_1,\dots,g_n}, 
    \qquad\text{with}\quad
    g_i= \tuple{\itp{`q},(\itp{`f}\,;\,\proj[n]{i})}~;~\case~.
  \end{displaymath}

  If $i`:J$,\quad $\itp{`f}\,;\,\proj[n]{i}=\itp{u_i}$ and then
  $g_i=\tuple{\itp{`q},\itp{u_i}}\,;\,\case$\quad which is~$f_i$.

  If $i`;J$, then\quad $\itp{`f}\,;\,\proj[n]{i}~=~\term[D^k]`*\fail$.
  Hence
  \begin{center}
    \begin{tikzpicture}
      \matrix[column sep=3em,row sep=.5em]{
        \node (a1) {$g_i=D^k$}; &
        \node (a2) {$\pdt`*\uno$}; &
        \node (a3) {$\pdt`*D$}; &
        \node (a4) {$D$}; & 
        \\  
        \node (b1) {$\quad=D^k$}; &
        \node (b2) {$\pdt`*\uno$}; &
        \node (b3) {$\uno$}; &
        \node (b4) {$D$}; & 
        \node (b5) {(by \drefFail)}; \\  
        \node (c1) {$\quad=D^k$}; &&
        \node (c3) {$\uno$}; &
        \node (c4) {$D$}; & 
        \\  
      };
      \draw[->,to right](a1) to node[above] {$\scriptstyle
        \tuple{\itp{`q},\term{D^k}}$} (a2);
      \draw[->,to right](a2) to node[above] {$\scriptstyle\id[\pdt]`*\fail$} (a3);
      \draw[->,to right](a3) to node[above] {$\scriptstyle\case$} (a4);
      \draw[->,to right](b1) to node[above] {$\scriptstyle
        \tuple{\itp{`q},\term{D^k}}$} (b2);
      \draw[->,to right](b2) to node[above] {$\scriptstyle\proj{2}$} (b3);
      \draw[->,to right](b3) to node[above] {$\scriptstyle\fail$} (b4);
      \draw[->,to right](c1) to node[above] {$\scriptstyle\term{D^k}$} (c3);
      \draw[->,to right](c3) to node[above] {$\scriptstyle\fail$} (c4);
    \end{tikzpicture}
  \end{center}
  So $g_i=f_i$ for any $i\leq n$,\quad and \quad $\tuple{\itp{`q},\itp{`f}}~;~\comp~~=~~\itp{`q`o`f}$.
  \qed
\end{proof}

We also need the standard following lemmas.

\begin{lemma}[Contextual rules]
  \label{lem:ctx-chg}%~\\
%   \begin{description}
%   \item[Exchange] 
%     Let\/ $`G=\{x_1,\dots,x_k\}$ and~$`s$ a substitution 
%     over~$\llbracket   1..k\rrbracket$.
%     Write $`s(`G)=\{`s(1),\dots,`s(k)\}$.
%     Then, for any term~$t$ whose free variables are in~$`G$,
%     \begin{displaymath}
%       \itp{t}\ =\ \tuple{\proj[k]{`s(1)},\dots,\proj[k]{`s(k)}}\;;\;
%       \itp[`s(`G)]{t}~.
%     \end{displaymath}
%   \item[Weakening]
%   Let $`G=\{x_1,\dots,x_k\}$ containing all free variables of a     % term~$t$, and~$y`;`G$.
%   Then
%   \begin{displaymath}
%     \itp[`G,y]{t} =\ 
%     \tuple{\proj[k+1]{1},\dots,\proj[k+1]{k}}\;;\;\itp{t}~.
%   \end{displaymath}  
%   \end{description}
  \textsl{Exchange: }
  Let\/ $`G=\{x_1,\dots,x_k\}$ and~$`s$ be a substitution 
  over~$\llbracket   1..k\rrbracket$.
  Write $`s(`G)=\{`s(1),\dots,`s(k)\}$.
  Then, for any term~$t$ whose free variables are in~$`G$,
  \begin{displaymath}
    \itp{t}\ =\ \tuple{\proj[k]{`s(1)},\dots,\proj[k]{`s(k)}}\;;\;
    \itp[`s(`G)]{t}~.
  \end{displaymath}
  \textsl{Weakening: }
  Let $`G=\{x_1,\dots,x_k\}$ containing all free variables of a     term~$t$, and~$y`;`G$.
  Then
  \begin{displaymath}
    \itp[`G,y]{t} =\ 
    \tuple{\proj[k+1]{1},\dots,\proj[k+1]{k}}\;;\;\itp{t}~.
  \end{displaymath}
\end{lemma}
\begin{lemma}[Substitution]
  \label{lem:subst}
  Given~$`G=\{x_1,\dots,x_k\}$, and two terms $t$ and~$u$ such that   $\fv{u}``(=`G$ and $\fv{t}``(=`G`U\{y\}$,
  \begin{center}
    \itp{t\subs[y]{u}} = 
    \begin{tikzpicture}[baseline=(a.base)]
      \node (a) {$D^k$};
      \node[right=of a] (b) {$D^k`*D$};
      \node[right= of b] (c) {$D^{k+1}$};
      \node[right= of c] (d) {$D$};
      \draw[->] (a) to 
      node[above] {$\scriptstyle\tuple{\id[],\itp{u}}$} (b);
      \draw[->] (b) to 
      node[above] {$\scriptstyle\eq$} (c);
      \draw[->] (c) to 
      node[above] {$\scriptstyle\itp[`G,y]{t}$} (d);    
    \end{tikzpicture}
  \end{center}
\end{lemma}

The soundness theorem is then a direct corollary of the following proposition, that is proved (in appendix~\ref{prf:snd}) by structural induction:
\begin{proposition}
  \label{prop:sound}
  If~$\M=(\C,D,\lam,\app,(\fc[{\co[c_i]}])^n_{i=1},\case,\fail)$ is a \lc model, then for any~$`G=\{x_1,\dots,x_k\}$ and any terms~$t_1,t_2$ such that $\fv{t_1}``(=`G$ and $t_1→t_2$, the interpretation given in Fig.~\ref{fig:itp-cat} satisfies
  \begin{math}
    \itp{t_1}=\itp{t_2}
  \end{math}.
\end{proposition}

%%% Local Variables: 
%%% mode: latex
%%% TeX-master: "catLC"
%%% End: 

\section{Completeness}
\label{sec:per}

In this part we shall prove that the converse of Theo.~\ref{theo:sound} holds in absence of match failure.
Namely if two terms have the same interpretation in any \lc-model then they are convertible using the rules of the calculus.
It means that, without match failure, the diagrams of
Fig.~\ref{fig:com-diag} are minimal.
% if two (hereditarily defined) terms have the same denotation in any
% \lc-model, then they are \lc-equivalent.

\begin{theorem}[Completeness]
  \label{theo:compl}
  If~$t$ and~$t'$ are two hereditarily defined \lc-terms such that in any categorical \lc-model \itp[]{t}=\itp[]{t'}, then 
$$t\eqlc t'~.$$
\end{theorem}

Notice that this theorem does not hold for undefined terms.
Indeed, every match failure receives the same denotation~\fail in any \lc-model, even though they are not \lc-convertible.
The completeness result is established using the same method as~\cite{FauMiq09}:
\begin{enumerate}
\item We define \PER, the Cartesian closed category of partial
  equivalence relation compatible with~\eqlc.
\item In this syntactic category, we construct a \lc-model~\Ms.
\item Then we show that if $\itp[]{t}=\itp[]{t'}$ in~\Ms, then 
  $t\eqlc t'$.
\end{enumerate}

\subsection{Partial equivalence relations}
\label{sec:pers}

\emph{Partial equivalence relations} (PER) %(\Def.~\ref{dfn:per}) 
are commonly used to transform a model of the untyped lambda calculus into a model of the typed lambda-calculus~\cite{Mitchell86,TanCoq87}.
Yet we use them here to instantiate the definition of \lc-models in the category of PER on \lc-terms.
Thereby we construct a syntactic model of the untyped \lc-calculus.

\begin{definition}[\per]
  \label{dfn:per}
  Given a set~$X$, a \emph{partial equivalence relation} on~$X$ is a binary relation~$R$ that is symmetric and transitive.
  We may write $\rel[R]{x}{y}$ instead of $(x,y)`:R$.
  A \per\ is a partial equivalence relation~$R$ on~\terms\ (the set of all \lc-terms) that is \emph{compatible} with~\lc-equivalence, which means:
  \begin{displaymath}
    \left\{
      \begin{array}[c]{c}
        \rel[R]{t}{t'}\\
        t_0\eqlc t'
      \end{array}
    \right. \quad\text{implies}\quad
    \rel[R]{t}{t_0}
  \end{displaymath}
\end{definition}

We write \class[R]{e} the equivalence class of an element~$e$ modulo~$R$ (or simply~\class{e} when it raises no ambiguity), and if it is non empty we say that~$e$ is \emph{accessible} by~$R$.
This is denoted by~$e`:R$.
We call the \emph{domain} of~$R$ (denoted by \dom{R}) the set of all its accessible elements modulo~$R$:
\begin{math}
  \dom{R}=\{\ \class[R]{e}\,/\, e`:R\ \}
\end{math}.
Notice that if a partial equivalence relation~$R$ is compatible with~\lc\ then by definition
\begin{equation}
  \label{eq:compat-class}
  t\eqlc t' \quad\implies\quad \class[R]{t}=\class[R]{t'}.
\end{equation}
It is well known that the family of partial equivalence relations can be provided with the usual semantic operators (arrow, and product) and constitute a \ccc~\cite[Theo~7.1]{Scott76}
To this end, we use the well-known Church's encoding for tuples:
\begin{displaymath}
  \begin{array}[t]{cclr}
    \pair[k]{x_1,\dots,x_k} &=& `lf.f\,x_1\dots x_k &\\
    \proj[k]{i} &=& `lp.p\,(`lx_1\dots x_k.x_i) & \quad \scriptstyle (i`:\llbracket 1..k \rrbracket) \\
  \end{array}
\end{displaymath}
(We may write $\pair{x,y}$ for $\pair[2]{x,y}$ and $\proj{i}$ for $\proj[2]{i}$).
It satisfies the expected equivalence:

\hfil
\begin{math}
  \proj[k]{i}\ \pair[k]{t_1,\dots ,t_k} \ \eqlc\ t_i 
\end{math}.

\begin{proposition}[Operations on {\per}s]
  Let $(R_i)_{1\leq i\leq n}$ be a family of PERs (with $n\geq 2$).
  Define $R_1→R_2$ and $R_1`*\dots`*R_n$ by
  \begin{displaymath}
    \begin{array}{l@{\quad\text{when}\quad}l}
      \rel[R→R']{t}{t'} & 
      \text{for any }u,u',\ \rel[R]{u}{u'} \implies \rel[R']{tu}{t'u'}
      \\
      \rel[R_1`*\dots`*R_k]{t}{u} & 
      \text{for each } i`:\llbracket 1..k\rrbracket,\quad
      \rel[R_i]{\proj[k]{i}t}{\proj[k]{i}u}      
    \end{array}
  \end{displaymath}
  Then if all the $R_i$'s are {\per}s, so are $R_1→R_2$ and     $R_1`*\dots`*R_n$.
\end{proposition}

\paragraph{The category \PER.}
The previous proposition enables providing the category of {\per}s with the structure of a \ccc.
In the category~\PER, objects are the {PER}s compatible with~\lc, and given two {\per}s~$A$ and~$B$ the morphisms of $A→B$ are the equivalence classes in \dom{A→B}.
The identity morphism on~$A$ is~\class[A→A]{`lx.x}, and the composition of~$\class{t}:A→B$ and~$\class{t'}:B→C$ is~$\class{t};\class{t'}=\class[A→C]{`lz.t'(tz)}$.
This defines correctly a category, as the composition is associative and has identity morphisms as neutral elements.
\\
\\
\parbox[b]{.7\linewidth}{
  ~\quad
  The categorical product of two {\per}s $A$ and~$B$ is   $(A`*B,~\class[A`*B→A]{\proj{1}},~\class[A`*B→B]{\proj{2}})$, and for~$\class{t}:C→A$ and~$\class{t'}:C→B$, the pairing of~\class{t_1} and~\class{t_2} is \linebreak[4]
  \begin{math}
    \tuple{\class{t},\class{t'}}=\class[C→A`*B]{`lx.\pair{tx,t'x}}~.
  \end{math}
  It is well defined (in particular it does not depend on the representative that we chose in the equivalence classes~\class{t} and~\class{t'}) and is universal for the diagram on the right.
}
\hfill
\parbox[b]{.25\linewidth}{\qquad
  \begin{tikzpicture}[baseline=(titi.south),node distance=4em]
    \node (c) {$C$};
    \node [below of= c] (p) {$A`*B$};
    \node [left of= p] (a) {$A$};
    \node [right of= p] (b) {$B$};
    \draw[->,dashed] (c) to 
    node[below right,tiny m] {\tuple{\class{t},\class{t'}}} (p);
    \draw[->] (c) to 
    node[above left] {$\scriptstyle \class{t}$} (a);
    \draw[->] (c) to 
    node[above right] {$\scriptstyle \class{t'}$} (b);
    \draw[->] (p) to 
    node[below] (titi) {$\scriptstyle \class{\proj{1}}$} (a);
    \draw[->] (p) to 
    node[below] {$\scriptstyle \class{\proj{2}}$} (b);
  \end{tikzpicture}
}
\\
The terminal object is the maximal \per~$\uno=\terms`*\terms$.
\\
The exponent of~$A$ and~$B$  is $B^A=A→B$, and the corresponding evaluation morphism is~$\ev=\class[B^A`*A→B]{`lx.(\proj{1}x)(\proj{2}x)}$~.
\\
\parbox[t]{0.35\linewidth}{
  \begin{tikzpicture}[node distance=3em and 4em, baseline=(a.base)]
    \node (a) {$C`*A$};
    \node [right= of a] (b) {$B$};
    \node [below= of a] (e) {$B^A`*A$};
    \draw[->] (a) to 
    node[above] (titi) {$\scriptstyle \class{t}$} (b);
    \draw[->] (a) to 
    node[left] {$\scriptstyle `L(\class{t})`*\id$} (e);
    \draw[->] (e) to 
    node[below right] {$\scriptstyle \ev$} (b);
  \end{tikzpicture}
}
\parbox[t]{0.65\linewidth}{
  The curried form of a morphism $\class{t}:C`*A→B$ is then
  \begin{math}
    `L(\class{t})= \class[C→B^A]{`lx.`ly.t\,\pair{x,y}} 
  \end{math}~.
  It is well defined and is the unique morphism that makes the diagram on the left commute.
}

\begin{proposition}
  \label{prop:per-ccc}
  \PER\ is a Cartesian closed category.
\end{proposition}

\subsection{Syntactic model in \PER.}

We will now define a \lc-model in the \ccc~\PER.
In this category, there is a trivial reflexive object, that is actually equal to its object of functions (as proved in appendix~\ref{prf:cpl-per}).
\begin{lemma}
  \label{lem:refl-per}
  Let~\D\ be the object~\eqlc\ in~\PER. Then $D=D^D$.
\end{lemma}

Also \eqlc\ is the object of~\PER that will be used to interpret untyped \lc-terms.
We do not need to define~\lam\ and~\app, and the morphisms~$\fc_i$'s and \case\ are quite intuitive:
informally, \fc\ is the constant function returning~\co, and \case\ takes an argument~$(`q,t)$ in~$\pdt`*\D$ and return~\cas{t}.
In the same way,~\fail\ is just a constant function returning a match failure (we arbitrarily choose one of the possible ones).
This actually defines a \lc-model (appendix~\ref{prf:cpl-per}).

\begin{definition}[Syntactic model]
  The syntactic model (or \emph{PER model}) of the \lc-calculus is
  $\Ms=(\PER,D,\id[D],\id[D],
  (\fc_i)_{1\leq i\leq n},\case,\fail)$,
  where:
  \begin{itemize}
  \item $D$ is the relation~\eqlc.
  \item given $\co$ a constructor, \fc\ is \class[\uno\D]{`lx.\co}.
  \item \case\ is \class[(\pdt`*D)\to D]{`lx.\cas[{(\co[c_i]\mapsto
        \proj[n]{i}(\proj{1}x))_{1\leq i\leq n}}]{\proj{2}x}}.
  \item \fail\ is \class[\uno\to D]{`lx.\cas[\ ]{\co[c_1]}}.
  \end{itemize}
\end{definition}

\begin{proposition}
  \label{prop:ms-model}
  \Ms\ is a \lc-model.
\end{proposition}

\paragraph{Case-binding completion.}
Remember that \lc-models do not distinguish different match failures (as a matter of fact, all of them are interpreted by~\fail).
That is because the interpretation of a term first ``completes'' each case-binding with branches~$\co_j\mapsto\fail$ if~$\co_j$ is not in its domain (\cf~the description of the denotation of a case-binding page~\pageref{pg:itp-cb}).
Also in the PER model, undefined terms are ``unblocked'' and the rule \rul{CaseCons} can be performed (and give $\cas[\ ]{\co[c_1]}$).
% That is the reason why the completeness theorem %(Theo.~\ref{theo:compl}) is restricted to hereditarily defined terms.
Now we formalise the idea of case-binding completion.
This enables an explicit definition of the interpretation of a term in the PER model, so that we can prove the completeness theorem.

\begin{definition}[Case-completion]
  \label{def:case-compl}
  The case-completion~\cpl{t} of a term~$t$ is defined by induction:

  \hfil
  \begin{math}
    \begin{array}[t]{c!{=}l!{\qquad}c!{=}c!{\qquad}c!{=}l}
      \cpl{x} & x &  \cpl{`lx.t} & `lx.\cpl{t} &
      \cpl{\cas{t}} & \cas[\cpl{`q}]{\cpl{t}}\\
      \cpl{\co} & \co & \cpl{tu} & \cpl{t}\cpl{u} &
      \multicolumn{2}{c}{ }\\
      \cpl{`q} &
      \{\co_i\mapsto u'_i/ 1\leq i\leq n\} &
      \multicolumn{4}{r}{\text{with }
        u'_i=
        \left\{
          \begin{array}[c]{cl}
            \cpl{u_i}&\text{ if }\co_i\mapsto u_i`:`q\\  
            \cas[\ ]{\co_1}&\text{ if }\co_i`;\dom{`q}
          \end{array}
        \right.}
    \end{array}
  \end{math}
  \hfil
\end{definition}

\begin{fact}
  \label{fact:cpl-equal}
  This case-completion does not unify different defined terms:
  if two defined terms have the same case-completion, then they are equal.
\end{fact}

\begin{proposition}
  \label{prop:def-itp}
  In the model~\Ms, the interpretation of a term~$t$ in a context
  $`G=x_1;\cdots;x_k$ is 

  \hfill
  $\itp{t}=\class[D^k\to D]{`lx.\cpl{t}\subs[x_i]{\proj[k]{i}x}}$
  \hfill
  (with~$x$ fresh in~$t$).
\end{proposition}

\subsection{Completeness result.}
\label{sec:compl-res}

The proposition~\ref{prop:def-itp} ensures that if two \lc-terms have the same denotation in the PER model, then they have the same case-completion \textit{modulo}~\D\ (\ie\ they are \lc-convertible).
It does not necessarily means that the two terms are \lc-equivalent themselves, as it is not true for match failure:
% \begin{displaymath}
%   \begin{array}[t]{c@{\ =\ }c@{=\ }c}
%   \itp[]{\cas[\co_1\mapsto`ly.yy]{\co_2}} &
%   \class[\uno→\D]{`lx.\cas[\co_1\mapsto`ly.yy;
%     \co_2\mapsto\cas[]{\co_1}]{\co_2}} &
%   \class[\uno→\D]{`lx.\cas[]{\co_1}}
%   \\
%   \itp[]{\cas[\co_2\mapsto`ly.y]{\co_1}} &
%   \class[\uno→\D]{`lx.\cas[\co_1\mapsto\cas[]{\co_1};
%     \co_2\mapsto`ly.y]{\co_1}} &
%   \class[\uno→\D]{`lx.\cas[]{\co_1}}
%   \end{array}
% \end{displaymath}
\begin{displaymath}
  \begin{array}[t]{c@{\ =\ }c@{\ \eqlc\ }c}
  \cpl{\cas[\co_1\mapsto`ly.yy]{\co_2}} &
  \cas[\co_1\mapsto`ly.yy;\co_2\mapsto\cas[]{\co_1}]{\co_2} &
  \cas[]{\co_1}
  \\
  \cpl{\cas[\co_2\mapsto`ly.y]{\co_1}} &
  \cas[\co_1\mapsto\cas[]{\co_1};\co_2\mapsto`ly.y]{\co_1} &
  \cas[]{\co_1}
  \end{array}
\end{displaymath}
Nevertheless, $\cas[\co_1\mapsto`ly.yy]{\co_2} \not\eqlc \cas[\co_2\mapsto`ly.y]{\co_1}$.
This explains why match failure all have the same interpretation in~\Ms.
However, this defect is restricted to undefined terms.
Now we show that the case-completion does not modify the \lc-equivalence on defined terms.
\begin{proposition}
  \label{prop:completion-eq}
  Let~$t_1$ and~$t_2$ be two hereditarily defined terms.
  Then
  \begin{displaymath}
    \cpl{t_1}\eqlc\cpl{t_2}
    \quad\implies\quad t_1\eqlc t_2
  \end{displaymath}
\end{proposition}

The proof of this proposition uses rewriting techniques, and relies on several lemmas (whose proofs are given in appendix~\ref{prf:cpl-rewrit}).
For technical reasons, we need to separate the rule \rul{CaseCase} from the other ones.
Also we write~\lcm\ the calculus with all the rules \textit{except} \rul{CaseCase}, and \textsc{cc} the rule \rul{CaseCase}.
\begin{fact}
  \label{fact:red-cpl}
  The definition of case-completion (Def.~\ref{def:case-compl})     preserves all \lc-redexes.
  Also if $t→u$ then~$\cpl{t}→\cpl{u}$, and if $\cpl{t}$ is a normal   form then so is~$t$.
\end{fact}
\begin{lemma}[Reduction on completed terms]%~\\
  \label{lem:cpl-red}
  \begin{enumerate}
  \item\label{it:cpl-red-lcm} 
    Let~$t$ be a defined term.\\
    Then, for any term~$t'$,
    \begin{center}
      $\cpl{t}→_{\lcm}t'$\quad implies\quad $t'=\cpl{t_0}$\; for    
      some~$t_0$ such that $t→t_0$.
    \end{center}
  \item \label{it:cpl-red-cc} 
    For any terms~$t, t'$,
    \begin{center}
      $\cpl{t}\tocc t'$\quad implies\quad $t'\tocc^*\cpl{t_0}$\; for   some~$t_0$ such that $t\tocc t_0$.
    \end{center}
  \end{enumerate}
\end{lemma}

The rule \rul{CaseCase} does not have the same behaviour as the other rules \wrt\ case-completion, and requires a special attention.
It has been proved that the reduction rule \rul{CaseCase} forms a confluent~\cite[Theo.~1]{AAA06} and strongly normalising~\cite[Prop.~2]{AAA06} rewriting system.
So every \lc-term~$t$ has a unique normal form~\cnf{t} for the rule \rul{CaseCase}.
It is characterised by the following equations:
\begin{displaymath}
  \begin{array}[t]{r!{=}l!{\qquad}r!{=}c}
    \cnf x & x & 
    \cnf{\{\co_i \mapsto u_i\,/\, {\scriptstyle  i`:I}\}} &
    \{\co_i \mapsto \cnf{u_i}\,/\, {\scriptstyle i`:I}\} \\
    \cnf \co & \co & 
    \multicolumn{2}{l}{
      \text{If } t = x~|~\co~|~`lx.u~|~t_1t_2~,
      \text{ then}}\\
    \cnf{`lx.t} & `lx. \cnf t &
    \cnf{\cas{t}}& \cas[\cnf{`q}]{\cnf{t}} \\ 
    \cnf{(tu)} & \cnf t \cnf u 
    &
    \cnf{\big(\cas{\cas[`f]{t}}\big)} &
    \cnf{(\cas[`q`o`f]{t})}\\
  \end{array}
\end{displaymath}

%%%%%%%%%%%%%%%%%%%%%%%%%%%%%% ICI!!! %%%%%%%%%%%%%%%%%%%%%%%%%%%%%%

\begin{lemma}
  \label{lem:cnf-cpl}
  Commutation case-completion/\rul{cc}-normal form\\
  For any term~$t$, $$\cnf{(\cpl{t})}=\cpl{\cnf{t}}~.$$
\end{lemma}

\begin{lemma}
  \label{lem:lcm-cnf}
  For any terms~$t,t'$, if $t→_{\lcm}t'$ then there exists a term~$u$ such that
  \begin{displaymath}
    \cnf{t}~→_{\lcm}^*~u~\tocc^*~\cnf{t'}~.
  \end{displaymath}
\end{lemma}

\begin{corollary}
  \label{cor:cpl-red}
  If~$t$ is hereditarily defined, then for any~$t'$, 
  \begin{displaymath}
    \cpl{t}→^*t' \quad\text{implies}\quad
    \cnf{t'}=\cpl{t_0} \text{ for some } t_0
    \text{ such that }
    t→^*t_0~.
  \end{displaymath}
\end{corollary}
\begin{proof}
  By induction on the reduction~$\cpl{t}→^*t'$.\\
  If $\cpl{t}=t'$, take $t_0=~\cnf{t}$.
  Now assume $\cpl{t}→^*u→_{R}t'$.
  By induction hypothesis, there is some~$u_0$ such that   $\cnf{u}=\cpl{u_0}$ and $t→^*u_0$.
  If~$u$ reduces on~$t'$ with the rule~$R=\rul{CaseCase}$, then   $\cnf{t'}=\cnf{u}=\cpl{u_0}$, and $t_0=u_0$ does the job.
  Otherwise, $\cpl{t}→^*u→_{\lcm}t'$.
 
  \vspace{-15pt}
  \hfil
%  \begin{center}
    \begin{tikzpicture}[baseline=(st.base)]
      \matrix[row sep=10pt,column sep=15pt]{
        \node (t) {\cpl{t}}; &&&& \node (u) {$u$}; && \node (t') {$t'$};\\
        &&&&&&\\
        &&&& \node (u0) {\cnf{u}=\cpl{u_0}}; & \node (u1) {\cpl{u_1}}; & \node (t0) {\cnf{t'}=\cpl{\cnf{u_1}}}; \\
        \node (st) {$t$}; &&&& \node (su0) {$u_0$}; & \node (su1) {$u_1$}; & \node (st0) {$\cnf{u_1}$};\\
      };
      \draw[->,to right] (t) to node[etoile] {} (u);
      \draw[->,to right] (u) to node[above,tiny m] {\lcm} (t');
      \draw[->,to right] (u0) to node[above,tiny m] {{\lcm}^*} (u1);
      \draw[->,to right] (u1) to node[above,tiny m] {\rul{cc}^*} (t0);
      \draw[->,dotted] (u) to node[m,left]{\rul{cc}^*}  (u0);
      \draw[->,dotted] (t') to node[m,right]{\rul{cc}^*} (t0);
      \draw[->,to right] (st) to node[etoile] {} (su0);
      \draw[->,to right] (su0) to node[etoile] {} (su1);
      \draw[->,to right] (su1) to node[at end, below,tiny m]       {\rul{cc}^*} (st0);
      \draw[->,serpent] (st) to (t);
      \draw[->,serpent] (su0) to (u0);
      \draw[->,serpent] (st0) to (t0);
    \end{tikzpicture}
%  \end{center}
    
    \noindent
    First of all, $u→_{\lcm}t'$ implies   $\cnf{u}~→_{\lcm}^*~u'~\tocc^*~\cnf{t'}$ for some~$u'$ (Lem.~\ref{lem:lcm-cnf}). \linebreak[4]
  Also $\cpl{u_0}~→_{\lcm}^*~u'$, and thus $u'=\cpl{u_1}$ for some term~$u_1$ such that $u_0→_{\lcm}^*u_1$ (Lem.~\ref{lem:cpl-red}.\ref{it:cpl-red-lcm}, since~$u_0$ is defined).
  Moreover, $\cpl{u_1}\tocc^*\cnf{t'}$ implies that~$\cnf{t'}$ is the \rul{CaseCase} normal form of~$\cpl{u_1}$.
  Hence $\cnf{t'}~=~\cnf{\cpl{u_1}}~=~\cpl{\cnf{u_1}}$ (by   Lem.~\ref{lem:cnf-cpl}).
  Also we can chose $t_0=\cnf{u_1}$.
  \qed
\end{proof}

Now we have all the ingredients we need to prove that the case-completion  preserves the \lc-equivalence on hereditarily defined terms.

\begin{proof} (of Prop.~\ref{prop:completion-eq}).
\\ \indent
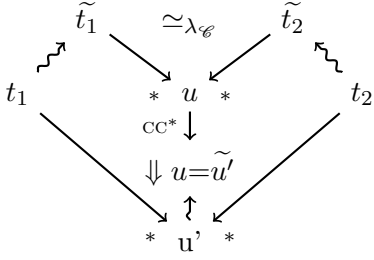
  \begin{tikzpicture}[baseline=(tt1.base)]
  \matrix[row sep=10pt,column sep=10pt]{
    & \node (tt1) {\cpl{t_1}}; & \node (eq) {\eqlc}; &
    \node (tt2) {\cpl{t_2}}; & \\
    \node (t1) {$t_1$}; && \node (u) {$u$}; && \node (t2) {$t_2$};\\
    && \node (tu) {\cnf{u}=\cpl{u'}}; && \\
    && \node (u') {u'}; && \\
  };
  \draw[->] (tt1) to node[etoile,auto=right] {} (u);
  \draw[->] (tt2) to node[etoile] {} (u);
  \draw[->] (u) to node[left, tiny m] {\rul{cc}^*} (tu);
  \draw[->] (t1) to node[etoile,auto=right] {} (u');
  \draw[->] (t2) to node[etoile] {} (u');
  \draw[->,serpent] (t1) to (tt1);
  \draw[->,serpent] (t2) to (tt2);
  \draw[->,serpent] (u') to (tu);
\end{tikzpicture}
  \hfill
  \parbox[t]{0.5\linewidth}{
    Let~$t_1, t_2$ hereditarily defined such that  $\cpl{t_1}\eqlc\cpl{t_2}$.
    Since the \lc-calculus satisfies the Church-Rösser property, there  is a term~$u$ such that~$\cpl{t_1}→^*u$ and~$\cpl{t_2}→^*u$.

    Hence Cor.~\ref{cor:cpl-red} provides a term~$u'$ such that  $\cnf{u} = \cpl{u'}$, and $t_i→^*u'$ for each $i`:\{1,2\}$.
    Thus $t_1\eqlc u'\eqlc t_2$.
    \qed
  }
\end{proof}

Together with the explicit definition of the interpretation of a term in the PER-model, this gives the result of completeness of \lc-models for terms with no match failure.

\begin{corollary}[Completeness]
  \label{cor:compl-per}
  Let~$t_1$ and~$t_2$ be two hereditarily defined terms whose free variables are in~$`G=\{x_1,\dots,x_k\}$ such that $\itp{t_1}=\itp{t_2}$ in the syntactic model~\Ms, then~$t_1\eqlc t_2$.
\end{corollary}
\begin{proof}
  By Prop.~\ref{prop:def-itp}, if $t_1$ and $t_2$ have the same interpretation in~\Ms, it means that
  \begin{displaymath}
    \class[D^k\to D]{`lx.\cpl{t_1}\subs[x_i]{\proj[k]{i}x}} =
    \class[D^k\to D]{`lx.\cpl{t_2}\subs[x_i]{\proj[k]{i}x}}~.
  \end{displaymath}
  Hence   $\rel{(`lx.\cpl{t_1}\subs[x_i]{\proj[k]{i}x})~\pair[k]{x_1,\dots,x_k}~}{~(`lx.\cpl{t_2}\subs[x_i]{\proj[k]{i}x})~\pair[k]{x_1,\dots,x_k}~}$.
  Since~$D$ is the \lc-equivalence relation on terms, it means that  $\cpl{t_1}\eqlc\cpl{t_2}$, which entails $t_1\eqlc t_2$ by   Prop.~\ref{prop:completion-eq}.
  \qed\end{proof}

\textit{A fortiori} if two hereditarily defined terms have the same  interpretation in \emph{any} \lc-model then they are \lc-equivalent, since~\Ms\ is a \lc-model.
This achieves the proof of Completeness theorem~(Theo.~\ref{theo:compl}).
\\

Notice that the separation theorem for the lambda calculus with constructors~\cite[Theo.~2]{AAA06} specifies that two hereditarily defined terms are either \lc-equivalent or (weakly) separable.
So any terms that can be separated by this syntactic lemma are also semantically distinguished by our definition of model.
However a slight modification of this definition could allow to semantically separate more terms.
If, instead of having one fail constant~\fail we had one for each constructor (say $\fail_1, fail_2$ \etc), we could ``complete'' a case binding with the corresponding fail constant in each undefined branch.
This would enable keeping track of the constructor that raises the match failure.
For instance, $\cas[\co_1\mapsto`lx.x]{\co_2}$ would be denoted by~$\fail_2$ and $\cas[\co_1\mapsto`lx.x]{\co_3}$ by~$\fail_3$.
Only terms like $\cas[\co_1\mapsto`lx.x]{\co_2}$ and $\cas[\co_3\mapsto`lx.xx]{\co_2}$ would not be semantically separated.

%%% Local Variables: 
%%% mode: latex
%%% TeX-master: "catLC"
%%% End: 

\section*{Conclusion}
\label{sec:concl}

We have defined a notion of categorical model for the lambda calculus with constructors that is reasonably complex:
in addition to the usual axioms of a \ccc, it involves three morphisms (or family of morphisms) and the commutation of six simple diagrams.
We have also proved that this  categorical model is complete for terms with no match failure.

Still, completeness does not hold for match failures.
This is due to the way we interpret the case-bindings.
Since the denotation we give to them is a point of~\pdt, it requires to ``fill'' artificially every undefined branch of a case-binding.
A way to cope with this problem could be to first identify the domain
$I``(=\llbracket 1..n\rrbracket$ of a case-binding $`q=\{\co_i\mapsto u_i/i`:I\}$, and interpret it by the point $(u_i)_{i`:I}$ of $D^{n_I}$ (where $n_I$ is the cardinal of~$I$).
The object that represents case-bindings would then be the sum (the dual notion of product) $\sum_{I``(=\llbracket 1..n\rrbracket}D^{n_I}$.
However, the definition loses its relative simplicity and some difficulties arise to define the case composition.

\paragraph{Future work}
A natural question is now to find some concrete instances of the categorical model.
The PER model is one, but it would be of great interest to have some non syntactic models.
We could try to adapt the historically first model of the pure lambda calculus~\cite{Scott70}.
However there is no reason for the usual Scott's~$D_{\infty}$ domain to satisfy the commutation of our diagrams.
A first step could be to find out a domain equation to characterise the lambda calculus with constructors, and then solve it with Scott's technique.

An other issue is to define a categorical model for the \textit{typed} \lc-calculus~\cite{Petit11}.
This type system is rather complex, basically because of the reduction rule \rul{CaseApp} that transforms a sub-term that is \textit{a priori} a function into a sub-term that is \textit{a priori} a data-structure.
To deal with this difficulty (and also to enable the typing of variadic constructors), the type syntax includes an application construct and the type system uses sub-typing.
Also defining a typed categorical model for the lambda calculus with constructors probably requires a categorical definition of this type application, and a way to express categorically this sub-typing relation.

%%% Local Variables: 
%%% mode: latex
%%% TeX-master: "catLC"
%%% End: 

\bibliographystyle{plain}
\bibliography{bibli}

\appendix
\newpage
\section{Proof of Soundness}
\label{prf:snd}

\textbf{Proposition \ref{prop:sound}.}
{\itshape
  If~$\M=(\C,D,\lam,\app,(\fc[{\co[c_i]}])^n_{i=1},\case,\fail)$ is a \lc-model, then for any~$`G=\{x_1,\dots,x_k\}$ and any terms~$t_1,t_2$ such that $\fv{t_1}``(=`G$ and $t_1→t_2$, the interpretation given in Fig.~\ref{fig:itp-cat} satisfies
  \begin{math}
    \itp{t_1}=\itp{t_2}
  \end{math}.
}
\begin{proof}
  Let~$t_1,t_2$ be two \lc-terms such that $t_1→t_2$.
  We prove by induction on the structure of~$t_1$ that for any~$`G$
  containing all free variables of~$t_1$, $\itp{t_1}=\itp{t_2}$.
  If the reduction does not involve a head redex, we immediately
  conclude with induction hypothesis.
  So we consider all possible reductions in head position:
  \begin{itemize}
  \item $t_1=(`lx.t)\,u$\ \ and\ \ $t_2=t\subs{u}$.\\
    \itp{t_1} =
    \begin{tikzpicture}[baseline=(a.base)]
      \node (a) {$D^k$};
      \node[right=2.5 of a] (b) {$D`*D$};
      \node[right=1.5 of b] (c) {$D^D`*D$};
      \node[right= of c] (d) {$D$};
      \draw[->] (a) to 
      node[above] {$\scriptstyle\tuple{(`L(f_t);\lam),\itp{u}}$}(b);
      \draw[->] (b) to 
      node[above] {$\scriptstyle\app`*\id[D]$} (c);
      \draw[->] (c) to 
      node[above] {$\scriptstyle\ev$} (d);
    \end{tikzpicture}\\
    with $f_t=$
    \begin{tikzpicture}[baseline=(a.base)]
      \node (a) {$D^k`*D$}; 
      \node[right= 0.5 of a] (b) {$D^{k+1}$}; 
      \node[right= of b] (c) {$D$};
      \draw [->] (a) to 
      node[above] {\eq} (b);
      \draw [->] (b) to 
      node[above] {$\scriptstyle \itp[`G,x]{t}$} (c);
    \end{tikzpicture}~.
    Thus 
    \begin{displaymath}
      \begin{array}[t]{cclr}
        \itp{t_1} &=&
        \tuple{\id[D],\itp{u}}~;~(`L(f_t);\lam;\app)`*\id[D]~;~\ev 
        & \\
        &=&
        \tuple{\id[D],\itp{u}}~;~`L(f_t)`*\id[D]~;~\ev 
        & \drefRefl \\
        &=& \tuple{\id[D],\itp{u}}~;~f_t &
        (\text{Def. of exponential})\\
        &=& \itp{t\subs{u}}&(\text{Lem.~\ref{lem:subst}})
      \end{array}
    \end{displaymath}
  \item $t_1=`lx.tx$ (with $x`;\fv{t}$) and $t_2=t$. Then
    $\itp{t_1}=`L(f_{tx})~;~\lam$ \\
    where $f_{tx}=$
    \begin{tikzpicture}[baseline=(a.base)]
      \node (a) {$D^k`*D$}; 
      \node[right=0.5 of a] (b) {$D^{k+1}$}; 
      \node[right=1.5 of b] (c) {$D`*D$}; 
      \node[right= of c] (d) {$D^D`*D$}; 
      \node[right=0.5 of d] (e) {$D$};
      \draw [->] (a) to 
      node[above] {\eq} (b);
      \draw [->] (b) to 
      node[above]{$\scriptstyle\tuple{\itp[`G,x]{t},\itp[`G,x]{x}}$} (c);
      \draw [->] (c) to 
      node[above] {$\scriptstyle \app`*\id[D]$} (d);
      \draw [->] (d) to 
      node[above] {$\scriptstyle \ev$} (e);
    \end{tikzpicture}~.\\
    But $x`;\fv{t}$ implies
    $\itp[`G,x]{t}=\tuple{\proj[k+1]{1},\dots,\proj[k+1]{k}}\;;\;\itp{t}$
    by weakening property (Lem.~\ref{lem:ctx-chg}), and
    $\itp[`G,x]{x}=\proj[k+1]{k+1}$.\\
    So 
    $f_{tx} =$
    \begin{tikzpicture}[baseline=(a.base)]
      \node (a) {$D^k`*D$}; 
      \node[right=0.5 of a] (b) {$D^{k+1}$}; 
      \node[right=3 of b] (c) {$D^k`*D$}; 
      \node[right=2 of c] (d) {$D^D`*D$}; 
      \node[right=0.5 of d] (e) {$D$};
      \draw [->] (a) to 
      node[above] {\eq} (b);
      \draw [->] (b) to node[above] {$\scriptstyle\tuple{
          \tuple{\proj[k+1]{1},\dots,\proj[k+1]{k}},\proj[k+1]{k+1}}$}   (c);
      \draw [->] (c) to 
      node[above] {$\scriptstyle (\itp{t};\,\app)`*\id[D]$} (d);
      \draw [->] (d) to 
      node[above] {$\scriptstyle \ev$} (e);
      \draw[ad] (a) to
      node [above] {$\scriptstyle \id[D^k`*D]$} (c);
    \end{tikzpicture}.\\
    By uniqueness of the exponential, $`L(f_{tx}) = \itp{t};\,\app$, and
    $\itp{t_1}=\itp{t};\,\app;\,\lam=\itp{t}$ by~\drefRefl.
  \item $t_1=\cas{\co[c_i]}$ and $t_2=u_i$,\quad
    where $`q=\{\co[c_j]\mapsto u_j/{j`:J}\}$, with 
    $J``(=\llbracket 1..n\rrbracket$.\\
    Then  $\itp{t_1}=\tuple{\,\tuple{f_1,\dots,f_n}\,,\,\itp{\co[c_i]}\,}~;~\case$
    \quad with $f_j$ = 
    \begin{math}
      \left\{
        \begin{array}[c]{l}
          \itp{u_j}\qquad\text{if }j`:J\\
          \term[D^k];\fail\qquad\text{otherwise}
        \end{array}
      \right.
    \end{math}\\
    and $\itp{\co[c_i]}=\term[D^k];\fc[c_i]$.\\
    The following diagram commutes:
    \begin{tikzpicture}[baseline=(a0.base)]
      \matrix[anchor=south,column sep=4em,row sep=2em]
      {
        \node (a0) {$D^k$}; & 
        \node (a1) {$\pdt`*\uno$}; &
        \node (a2) {$\pdt`*D$};\\
        & \node (b1) {$\pdt$}; &
        \node (b2) {$D$};\\
      };
      \draw[->] (a0) to node[above] {
        $\scriptstyle\tuple{\,\tuple{f_1,\dots,f_n}\,,\,\term[D^k]\,}$
      } (a1);
      \draw[->] (a1) to node[above] 
      {$\scriptstyle \id[\pdt]`*\fc[c_i]$} (a2);
      \draw[->] (a0) to node[below left] 
      {$\scriptstyle \tuple{f_1,\dots\,,f_n}$} (b1);
      \draw[->] (b1) to node[below] 
      {$\scriptstyle\proj[n]{i}$} (b2);
      \draw[->] (a1) to node[left] {$\scriptstyle\eq$} (b1);
      \draw[->] (a2) to node[right] {$\scriptstyle\case$} (b2);
      \path (a1) to node {\drefCO} (b2);
    \end{tikzpicture},\\
    so $\itp{t_1}~=~\tuple{f_1,\dots\,,f_n}\;;\;\proj[n]{i}
    ~=~f_i~=~\itp{u_i}$.
  \item $t_1=\cas{(tu)}$\quad and\quad $t_2=(\cas{t})~u$.\\
    \begin{math}
      \itp{t_1} = \tuple{\,\itp{`q}\,,\,\itp{tu}}~;~\case
    \end{math}\quad
    with \; $\itp{tu}=\tuple{\itp{t},\itp{u}}~;~(\app`*\id[D])~;~\ev$\\
    \begin{math}
      \itp{t_2} = \tuple{
        \big( \tuple{\,\itp{`q}\,,\,\itp{t}}~;~\case \big)\,,\,
        \itp{u}
      }~;~(\app`*\id[D]);\ev
    \end{math}\\
    So \itp{t_1}~=~\itp{t_2} because the following diagram commutes:

    \begin{tikzpicture}
      \matrix[anchor=south,column sep=3em,row sep=1.5em]
      { & 
        \node (d) {$\pdt`*(D`*D)$}; &
        \node (d1){$\pdt`*(D^{D}`*D)$}; &
        \node (d2){$\pdt`*D$}; & \\
        \node (i) {$D^k$}; &&
        \node {\drefCA}; &&
        \node (f) {$D$}; \\ &
        \node (g){$(\pdt`*D)`*D$}; &
        \node (g1) {$D`*D$}; &
        \node (g2) {$D^D`*D$}; &\\
      };
      \draw [->](i) to
      node[above left] {$\scriptstyle\tuple{\,
          \itp{`q}\,,\,\tuple{\itp{t},\itp{u}}}$} (d);
      \draw [->](i) to
      node[below left] {$\scriptstyle\tuple{\,\tuple{
            \itp{`q}\,,\,\tuple{\itp{t}}},\itp{u}}$} (g);
      \draw [<->] (g) to
      node [label=left:{$\scriptstyle\eq$}] {} (d);
      \draw [->] (d) to node[label=above:{$\scriptstyle\id[]`*(\app`*\id[])$}] 
      {} (d1);
      \draw [->] (d1) to
      node [label=above:{$\scriptstyle\id[]`*\ev$}] {} (d2);
      \draw [->] (d2) to
      node [label=above right:{$\scriptstyle\case$}] {} (f);
      \draw [->] (g) to
      node [label=below:{$\scriptstyle\case`*\id[]$}] {} (g1);
      \draw [->] (g1) to
      node [label=below:{$\scriptstyle\app`*\id[]$}] {} (g2);
      \draw [->] (g2) to
      node [label=below right:{$\scriptstyle\ev$}] {} (f);
    \end{tikzpicture}
    %     % 

  \item $t_1=\cas{`lx.t}$ and $t_2=`lx.\cas{t}$\quad
    with $x`;\fv{`q}$.\\
    \begin{math}
      \itp{t_1}~= \tuple{\itp{`q},(`L(f_t);\lam)}~;~\case
    \end{math}\quad
    with $f_t=$
    \begin{tikzpicture}[baseline=(a.base)]
      \node (a) {$D^k`*D$}; 
      \node[right= 0.5 of a] (b) {$D^{k+1}$}; 
      \node[right= of b] (c) {$D$};
      \draw [->] (a) to 
      node[above] {\eq} (b);
      \draw [->] (b) to 
      node[above] {$\scriptstyle \itp[`G,x]{t}$} (c);
    \end{tikzpicture}, and
    \\
    \begin{math}
      \itp{t_2}~= `L(f_{\cas{t}});\lam
    \end{math}\quad
    with $f_{\cas{t}}=$
    \begin{tikzpicture}[baseline=(a.base)]
      \node (a) {$D^k`*D$}; 
      \node[right= 0.5 of a] (b) {$D^{k+1}$}; 
      \node[right= 1.2 of b] (c) {$\pdt`*D$}; 
      \node[right= 0.7 of c] (d) {$D$};
      \draw [->] (a) to 
      node[above] {\eq} (b);
      \draw [->] (b) to 
      node[above, tiny m] {\tuple{\itp[`G,x]{`q},\itp[`G,x]{t}}}(c);
      \draw [->] (c) to 
      node[above, tiny m] {\case} (d);
    \end{tikzpicture}.
    \\
    So
    \begin{math}
      \begin{array}[t]{c!{=}l!{\qquad}r}
        \itp{t_1} & \tuple{\itp{`q},(`L(f_t);\lam)}~;~\case & \\
        & \tuple{\itp{`q},`L(f_t)}~;~(\id[\pdt]`*\lam)~;~\case & \\
        & \tuple{\itp{`q},`L(f_t)}~;~\abstr{\case}~;~\lam &
        \text{by}~\drefCL\\
      \end{array}
    \end{math}
    
    Hence $\itp{t_1}=\itp{t_2}$ if
    \begin{math}
      \tuple{\itp{`q},`L(f_t)}~;~\abstr{\case} = `L(f_{\cas{t}})
    \end{math}.
    
    Remember that 
    $\abstr{\case}=`L(f_{\case})$, with\\
    $f_{\case}=$
    \begin{tikzpicture}[vcenter]
      \matrix[column sep=12pt]{
        \node[m] (a) {(\pdt`*D^{D})`*D}; &
        \node[m] (b) {\pdt`*(D^{D}`*D)}; &&&
        \node[m] (c) {\pdt`*D}; &&
        \node[m] (d) {D}; \\
      };
      \draw[->] (a) to node[above,tiny m] {\eq} (b);
      \draw[->] (b) to node[above,tiny m] {\id[\pdt]`*\ev} (c);
      \draw[->] (c) to node[above,tiny m] {\case} (d);
    \end{tikzpicture}.
    To simplify this equation, we use this intermediate lemma (that follows from the uniqueness of exponent).
    \begin{lemma}
%      \label{lem:compo-exp}
      In any \ccc, given four objects~$A,B,C$ and~$C'$, and three morphisms $g:C`*A→B$, $g':C'`*A→B$ and $h:C→C'$, $$`L(g)=h;`L(g') \iff g=(h`*\id[A]);g'~.$$
    \end{lemma}
    
    Thus $\itp{t_1}=\itp{t_2}$ 
    \quad if\quad
    \begin{math}
      (\tuple{\itp{`q},`L(f_t)}`*\id[d])~;~f_{\case} = f_{\cas{t}}
    \end{math}.
    
    Remark that 
    \begin{math}
      (\tuple{\itp{`q},`L(f_t)}`*\id[d])~;~f_{\case} =
      \text{lhs}~;~\case
    \end{math}, with\\
    lhs\hspace{-5pt}
    \begin{tikzpicture}[baseline=(a1.base)]
      \matrix[column sep=25pt]{
        \node[tiny m] (a1) {=~D^k\!`*D}; &&&&
        \node[tiny m] (c1) {\pdt\!`*(D^{D}\!`*D)}; &&
        \node[tiny m] (d1) {\pdt\!`*D}; \\
        \node[tiny m] (a2) {=~D^k\!`*D}; &&
        \node[tiny m] (b2) {\pdt\!`*(D^{k}\!`*D)}; &&
        \node[tiny m] (c2) {\pdt\!`*(D^{D}\!`*D)}; &&
        \node[tiny m] (d2) {\pdt\!`*D}; \\
        \node[tiny m] (a3) {=~D^k\!`*D}; &&
        \node[tiny m] (b3) {\pdt\!`*(D^{k}\!`*D)}; &&&&
        \node[tiny m] (d3) {\pdt\!`*D}; \\
      };
      \node[tiny m] (b1) at ($(a1)!0.6!(c1)$) {(\pdt\!`*D^{D})\!`*D};
      \draw[->,to right] (a1) to node[above,tiny m]
      {\tuple{\itp{`q},`L(f_t)}\!`*\!\id[D]} (b1);
      \draw[->,to right] (b1) to node[above,tiny m] {\eq} (c1);
      \draw[->,to right] (c1) to node[above,tiny m]
      {\id[\pdt]\!`*\ev} (d1);
      \draw[->,to right] (a2) to node[above,tiny m]
      {\tuple{(\proj{1}\,;\,\itp{`q}),\id[]}} (b2);
      \draw[->,to right] (b2) to node[above,tiny m] 
      {\id[\pdt]\!`*\!(`L(f_t)\!`*\!\id[D])} (c2);
      \draw[->,to right] (c2) to node[above,tiny m]
      {\id[\pdt]\!`*\ev}  (d2);
      \draw[->,to right] (a3) to node[above,tiny m]
      {\tuple{(\proj{1}\,;\,\itp{`q}),\id[]}} (b3);
      \draw[->,to right] (b3) to node[above,tiny m] 
      {\id[\pdt]`*f_t} (d3);
    \end{tikzpicture}

    On the other hand, $f_{\cas{t}}=\text{rhs}~;~\case$, with\\
    rhs\hspace{-5pt}
    \begin{tikzpicture}[baseline=(a1.base)]
      \matrix[column sep=12pt]{
        \node[tiny m] (a1) {=~D^k\!`*\!D}; &
        \node[tiny m] (b1) {D^{k+1}}; &
        \node[tiny m] (c1) {D^{k+1}`*D^{k+1}}; &&&&&&
        \node[tiny m] (e1) {\pdt\!`*\!D}; &\\
        \node[tiny m] (a2) {=~D^k\!`*\!D}; &
        \node[tiny m] (b2) {D^{k+1}}; &
        \node[tiny m] (c2) {D^{k+1}`*D^{k+1}}; &&&
        \node[tiny m] (d2) {D^{k}`*D^{k+1}}; &&&
        \node[tiny m] (e2) {\pdt\!`*\!D}; &
        \node (f2) {\hspace{-15pt}\tiny (Lem.~\ref{lem:ctx-chg})};\\
        \node[tiny m] (a3) {=~D^k\!`*\!D}; &&
        \node[tiny m] (b3) {(D^k\!`*\!D)`*(D^k\!`*\!D)}; &&&
        \node[tiny m] (d3) {D^{k}`*D^{k+1}}; &&&
        \node[tiny m] (e3) {\pdt\!`*\!D}; &\\
        \node[tiny m] (a4) {=~D^k\!`*\!D}; &&
        \node[tiny m] (b4) {(D^k\!`*\!D)`*(D^k\!`*\!D)}; &&&&&&
        \node[tiny m] (e4) {\pdt\!`*\!D}; &\\
      };
      \draw[->,to right] (a1) to node[above,tiny m] {\eq} (b1);
      \draw[->,to right] (b1) to node[above,tiny m]{\tuple{\id[],\id[]}} (c1);
      \draw[->,to right] (c1) to node[above,tiny m]{\itp[`G,x]{`q}`*\itp[`G,x]{t}} (e1);
      
      \draw[->,to right] (a2) to node[above,tiny m] {\eq} (b2);
      \draw[->,to right] (b2) to node[above,tiny m]{\tuple{\id[],\id[]}} (c2);
      \draw[->,to right] (c2) to node[above,tiny m]{\tuple{...,\proj[k+1]{k}}`*\id[]} (d2);
      \draw[->,to right] (d2) to node[above,tiny m]{\itp{`q}`*\itp[`G,x]{t}} (e2);
      
      \draw[->,to right] (a3) to node[above,tiny m]{\tuple{\id[],\id[]}} (b3);
      \draw[->,to right] (b3) to node[above,tiny m] {\proj{1}`*\eq}(d3);
      \draw[->,to right] (d3) to node[above,tiny m]{\itp{`q}`*\itp[`G,x]{t}} (e3);
      \draw[->,to right] (a4) to node[above,tiny m]{\tuple{\id[],\id[]}} (b4);
      \draw[->,to right] (b4) to node[above,tiny m]{(\proj{1}\,;\,\itp{`q})`*f_t} (e4);
    \end{tikzpicture}

    Finally
    \begin{math}
      \text{rhs} = \text{lhs} = 
      \tuple{(\proj{1}\,;\,\itp{`q})~,~f_t}
    \end{math},\quad
    and so \itp{t_1}=\itp{t_2}.
    % % 
  \item $t_1=\cas{\cas[`f]{t}}$\; and\; $t_2=\cas[`q`o`f]{t}$.\\
    \begin{math}
      \itp{t_1}~=   \big(\tuple{\itp{`q},\tuple{\itp{`f},\itp{t}}}\big)~;~
      (\id[\pdt]`*\case)~;~\case
    \end{math},\quad
    and\\
    \begin{math}
      \itp{t_2}=\big(\tuple{\itp{`q`o`f},\itp{t}}\big)~;~\case
    \end{math}.\\
    Both terms have the same interpretation if the following diagram  commute:
    
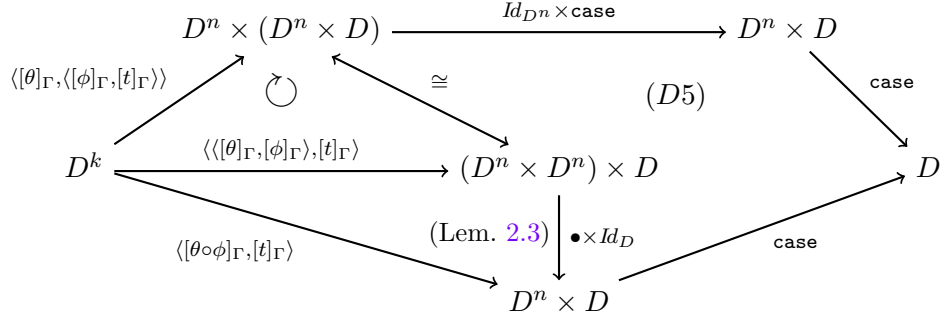
\begin{center}
      \begin{tikzpicture}[baseline=(a.base)]
        \matrix[column sep=2em,row sep=3em]{&
          \node (b) {$\pdt`*(\pdt`*D)$}; &&
          \node (c) {$\pdt`*D$}; & \\
          \node (a) {$D^k$}; &&
          \node (e) {$(\pdt`*\pdt)`*D$}; &&
          \node (d) {$D$}; \\
          && \node (f) {$\pdt`*D$}; &&\\
        };
        \draw[->] (a) to node[above left] {$\scriptstyle
          \tuple{\itp{`q},\tuple{\itp{`f},\itp{t}}}$} (b);
        \draw[->,to right] (b) to node[above] {$\scriptstyle
          \id[\pdt]`*\case$} (c);
        \draw[->] (c) to node[above right] {$\scriptstyle\case$}(d);
        \draw[->] (a) to node[below left] {$\scriptstyle
          \tuple{\itp{`q`o`f},\itp{t}}$} (f);
        \draw[->] (f) to node[below right] {$\scriptstyle\case$}(d);
        \draw[<->] (b) to node[above right] {$\scriptstyle\eq$} (e);
        \draw[->] (e) to node[right] {$\scriptstyle\comp`*\id[D]$} (f);
        \draw[->,to right] (a) to node[above] {$\scriptstyle
          \tuple{\tuple{\itp{`q},\itp{`f}},\itp{t}}$} (e);
        \path (e) to node (D) {\drefCC} (c);
        \node at (D -| b) {\LARGE$\circlearrowright$};
        \node at ($(b)!0.75!(f)$) {\small(Lem.~\ref{lem:cat-cc})};
      \end{tikzpicture}
    \end{center}
    The upper triangle commutes by uniqueness of the product,
    the triangle below commutes if~\drefFail\ commutes (consequence of Lem.~\ref{lem:cat-cc}), and the right part of the diagram is  exactly~\drefCC.
    Also the interpretation is correct \wrt\ {\sc CaseCase}  if~\drefCC\  and~\drefFail\ commute.
\qed
  \end{itemize}
\end{proof}

\section{Proofs for Completeness}
\label{prf:cpl}

\subsection{Some properties of \PER.}
\label{prf:cpl-per}

\noindent
\textbf{Lemma \ref{lem:refl-per}.}
{\itshape
    Let~\D\ be the object~\eqlc\ in~\PER. Then $D=D^D$.
}
\begin{proof}
  \begin{description}
  \item[$\subseteq$:] If~$\rel{t}{t'}$, then $\rel{u}{u'}$ implies
    $\rel{tu}{t'u'}$ by definition of~$D$.
    This means~$\rel[D^D]{t}{t'}$
  \item[$\supseteq$:] Assume $\rel[D^D]{t}{t'}$, and choose~$x$ not free in~$t$ nor~$t'$.
    Since $\rel{x}{x}$, then $\rel{tx}{t'x}$.
    So $\rel{`lx.tx}{`lx.t'x}$ by contextual closure, and $\rel{t}{t'}$ by~\rul{LamApp}.
    \qed
  \end{description}
\end{proof}

%%%%%%%%%%%%%%%%%%%%%%%%%%%%%%%%%%%%%%%%%%%%%%%%%%%%%%%%%%%%%%%%%%%%%%

\noindent
\textbf{Proposition \ref{prop:ms-model}.}
{\itshape
  Let $\Ms=(\PER,D,\id[D],\id[D],
  (\fc_i)_{1\leq i\leq n},\case,\fail)$,
  where:
  \begin{itemize}
  \item $D$ is the relation~\eqlc.
  \item given $\co$ a constructor, \fc\ is \class[\uno\D]{`lx.\co}.
  \item \case\ is \class[(\pdt`*D)\to D]{`lx.\cas[{(\co[c_i]\mapsto
        \proj[n]{i}(\proj{1}x))_{1\leq i\leq n}}]{\proj{2}x}}.
  \item \fail\ is \class[\uno\to D]{`lx.\cas[\ ]{\co[c_1]}}.
  \end{itemize}
  \Ms\ is a \lc-model.
}
\begin{proof}
  \PER\ is a Cartesian closed category by Prop.~\ref{prop:per-ccc}, and $\id[\D]$ is an isomorphism from $\D$ to~$D^\D$ by Lem.~\ref{lem:refl-per}.
  We first check that the morphisms are well-defined:
  \begin{itemize}
  \item $\fc`:\dom{\uno→D}$ for each~constructor~\co.
    Indeed, for any terms $u,u'$,
    \linebreak[4]
    \begin{math}
      (`lx.\co)~u~\eqlc~\co~\eqlc~(`lx.\co)~u'
    \end{math}.
    Hence \rel[\uno→D]{`lx.\co}{`lx.\co}.
    In the same way, $\fail`:\dom{\uno→D}$.
  \item $\case`:\dom{\pdt`*\D→\D}$ since
    \begin{math}
      `lx.\cas[{(\co_i\mapsto
        \proj[n]{i}(\proj{1}x))_{i=1}^n}]{\proj{2}x}
      `: (\pdt`*D)\to D
    \end{math}.
    Indeed, let $\rel[(\pdt`*D)]{t}{u}$.
    By definition, \rel[D]{\proj[n]{i}(\proj{1}t)}{\proj[n]{i}(\proj{1}u)}, and \rel[D]{\proj{2}t}{\proj{2}u}.
    Thus\\
    \begin{math}
      \begin{array}[t]{c!{\eqlc}l}
        \big( `lx.\cas[{(\co[c_i]\mapsto
          \proj[n]{i}(\proj{1}x))_{i=1}^n}]{\proj{2}x} \big)t
        & \cas[{(\co[c_i]\mapsto
          \proj[n]{i}(\proj{1}t))_{i=1}^n}]{\proj{2}~t}\\
        &\cas[{(\co[c_i]\mapsto
          \proj[n]{i}(\proj{1}u))_{i=1}^n}]{\proj{2}~u}\\
        & \big( `lx.\cas[{(\co[c_i]\mapsto
          \proj[n]{i}(\proj{1}x))_{i=1}^n}]{\proj{2}x} \big)u
      \end{array}
    \end{math}
  \end{itemize}
  Finally by Prop.~\ref{prop:equiv-cacl} it is sufficient to show that the diagrams \drefRefl,  \drefCO, \drefCA, \drefCC\ and \drefFail\ of Fig.~\ref{fig:com-diag} commute.
  For~\drefRefl\ it is obvious with $\lam=\app=\id[D]$.
  We show the commutation porperty for the other diagram.
  \begin{description}
  \item[{\drefCO:}]
    We show that $rhs~=~\proj[n]{i}$, where 
    rhs~=~$h_{\eq}~;~(\id[\pdt]`*\fc[\co_i])~;~\case$ 
    (with $h_{\eq}=\class[\pdt→\pdt`*\uno]{`lx.\pair{x,x}}$).
    Notice that 
    \begin{math}
      (\id[\pdt]`*\fc[\co_i]) =
      \class{`lx.\pair{\proj{1}x,(`lx.\co_i)(\proj{2}x)}}
    \end{math}~.
    We simplify~rhs, considering terms up to     \lc-equivalence~\eqref{eq:compat-class}.\\
    \begin{math}
      \begin{array}[t]{c!{=}l!}
        \text{rhs} & \class[\pdt→D]{`lz.t_{\case}\big(
          (`lx.\pair{\proj{1}x,(`lx.\co_i)x})~((`lx.\pair{x,x})z)            \big)} \\
        & \class[\pdt→D]{`lz.t_{\case}\big(
          \pair{\proj{1}{\pair{z,z}},(`lx.\co_i)(\proj{2}\pair{z,z})}
          \big)} \\
        & \class[\pdt→D]{`lz.t_{\case}(\pair{z,\co_i})} \\
        & \class[\pdt→D]{`lz.
          \cas[{(\co_i\mapsto\proj[n]{i}(\proj{1}\pair{z,\co_i})
            )_{i=1}^n}]{\proj{2}\pair{z,\co_i}}
        } \\
        & \class[\pdt→D]{`lz.
          \cas[{(\co_i\mapsto\proj[n]{i}(\proj{1}\pair{z,\co_i})
            )_{i=1}^n}]{\proj{2}\pair{z,\co_i}}
        } \\
        & \class[\pdt→D]{`lz.
          \cas[{(\co_i\mapsto\proj[n]{i}~z)_{i=1}^n}]{\co_i}} \\
        & \class[\pdt→D]{`lz.\proj[n]{i}~z} \qquad
        \text{by \rul{CaseCons}}\\
        & \proj[n]{i} \\
      \end{array}
    \end{math}
  \item[{\drefCA:}]
    We show that lhs~=~rhs, where
    \begin{math}
      \text{lhs} = (\case`*\id[D])~;~(\app`*\id[D])~;~\ev
    \end{math}, \\
    and 
    \begin{math}
      \text{rhs} = h_{\eq}~;~(\id[\pdt]`*(\app`*\id[D]))~;~
      (\id[\pdt]`*\ev);\case
    \end{math},
    with 
    
    \hfil
    $h_{\eq}=\class[(\pdt`*D)`*D→\pdt`*(D`*D)]{`lx.
      \pair{\proj{1}(\proj{1}x),\pair{\proj{2}(\proj{1}x),\proj{2}x}}}$.
    \\
    Notice that $\app`*\id[D]=\id[D`*D]$, so
    \begin{math}
      \text{lhs} = (\case`*\id[D])~;~\ev
    \end{math}, 
    and 

    \hfill
    \begin{math}
      \text{rhs} = h_{\eq}~;~(\id[\pdt]`*\ev);\case
    \end{math}.
    
    \begin{math}
      \begin{array}[t]{c!{=}l}
        \text{lhs} & \class{`lz.
          (`lx.(\proj{1}x)(\proj{2}x))~
          \big((`lx.\pair{t_{\case}(\proj{1}x),\proj{2}x})z\big)}\\
        &\class{`lz.(`lx.(\proj{1}x)(\proj{2}x))~
          \pair{t_{\case}(\proj{1}z),\proj{2}z}}\\
        & \class{`lz.(t_{\case}(\proj{1}z))~(\proj{2}z)}\\
        & \class{`lz.\big(
          \cas[{(\co_i\mapsto
            \proj[n]{i}(\proj{1}(\proj{1}z)))_{i=1}^n}]{\proj{2}(\proj{1}z)}\big)
          (\proj{2}z)}\\
      \end{array}
    \end{math}

    \begin{math}
      \begin{array}[t]{c!{=}l}
        \text{rhs} & \class{`lz. t_{\case}~
          (`ly.\pair{\proj{1}y,(`lx.(\proj{1}x)(\proj{2}x))(\proj{2}y)})~
          ((`lx.
          \pair{\proj{1}(\proj{1}x),\pair{\proj{2}(\proj{1}x),\proj{2}x}})z)}\\
        & \class{`lz. t_{\case}~
          (`ly.\pair{\proj{1}y,(\proj{1}(\proj{2}y))(\proj{2}(\proj{2}y))})~
          \pair{\proj{1}(\proj{1}z),\pair{\proj{2}(\proj{1}z),\proj{2}z}}}\\
        & \class{`lz. t_{\case}~
          \pair{\proj{1}(\proj{1}z),(\proj{2}(\proj{1}z))(\proj{2}z)}}\\
        & \class{`lz. \cas[{(\co_i\mapsto
            \proj[n]{i}(\proj{1}(\proj{1}z)))_{i=1}^n}]{
            \big(\proj{2}(\proj{1}z)~(\proj{2}z)\big)} }\\
        & \class{`lz. \big( \cas[{(\co_i\mapsto
            \proj[n]{i}(\proj{1}(\proj{1}z)))_{i=1}^n}]{
            \proj{2}(\proj{1}z)}\big)(\proj{2}z) }
        \qquad \text{by \rul{CaseApp}}\\
      \end{array}
    \end{math}
  \item[{\drefCC:}]
    Let lhs~=~$(\comp`*\id[D])~;~\case$,\quad and 
    rhs~=~$h_{\eq}~;~(\id[\pdt]`*\case)~;~\case$, with

    \hfil
    $h_{\eq}=\class[(\pdt`*\pdt)`*D→\pdt`*(\pdt`*D)]{
      `lx.\pair{\proj{1}(\proj{1}x)\,,\,
        \pair{\proj{2}(\proj{1}x)\,,\,\proj{2}x}}
    }$.
    \\
    Then \drefCC\ commutes means lhs~=~rhs.
    \\
    Remember that $\comp:\pdt`*\pdt→\pdt$ is the pairing of all     $(\id[\pdt]`*\proj[n]{i})~;~\case$.
    Thus

    \hfil
    \begin{math}
      \begin{array}[t]{r!{=}l}
        \comp &\class{`lx.\pair{\dots,
            (`ly.t_{\case}~\pair{\proj{1}y\,,\,\proj[i]{n}(\proj{2}y)})x,
            \dots}} \\
        & \class{`lx.\pair{\dots,
            t_{\case}~\pair{\proj{1}x\,,\,\proj[i]{n}(\proj{2}x)},
            \dots}} \\
        \comp`*\id[D] &
        \class{`lx. \pair{~\pair{\dots,
              t_{\case}~\pair{\proj{1}(\proj{1}x)\,,\,
                \proj[i]{n}(\proj{2}(\proj{1}x))},
              \dots}~,\,
            \proj{2}x\,}} \\
        \text{lhs} &
        \class{ `lz. t_{\case}~
          \pair{~\pair{\dots,
              t_{\case}~\pair{\proj{1}(\proj{1}z)\,,\,
                \proj[i]{n}(\proj{2}(\proj{1}z))},
              \dots}~,\,
            \proj{2}z\,}
        } \\
        & \class{ `lz. 
          \cas[{(\co_i\mapsto t_{\case}~\pair{\proj{1}(\proj{1}z)\,,\,
              \proj[i]{n}(\proj{2}(\proj{1}z))}
            )_{i=1}^n}]{\proj{2}z}
        } \\
        & \class{ `lz. 
          \cas[{(\co_i\mapsto t_{\case}~\pair{\proj{1}(\proj{1}z)\,,\,
              \proj[i]{n}(\proj{2}(\proj{1}z))}
            )_{i=1}^n}]{(\proj{2}z)}
        } \\
        & \class{ `lz. 
          \cas[{(\co_i\mapsto 
            \cas[{(\co_j\mapsto
              \proj[n]{j}(\proj{1}(\proj{1}z)))_{j=1}^n}]{
              (\proj[i]{n}(\proj{2}(\proj{1}z)))}
            )_{i=1}^n}]{(\proj{2}z)}
        } \\
      \end{array}
    \end{math}
    \\
    \begin{math}
      \begin{array}[t]{c!{=}l}
        \text{rhs} & \class{
          `lz.t_{\case}~\big((`lx.\pair{\proj{1}x\,,\,t_{\case}(\proj{2}x)})~
          \pair{\proj{1}(\proj{1}z)\,,\,
            \pair{\proj{2}(\proj{1}z)\,,\,\proj{2}z}}
          \big) }\\
        & \class{ `lz.t_{\case}~(
          \pair{\proj{1}(\proj{1}z)\,,\,t_{\case}
            \pair{\proj{2}(\proj{1}z)\,,\,\proj{2}z}}) }\\
        & \class{ `lz.
          \cas[{(\co_i\mapsto
            \proj[n]{i}(\proj{1}(\proj{1}z)))_{i=1}^n}]{
            t_{\case}\,\pair{\proj{2}(\proj{1}z)\,,\,\proj{2}z}}
        }\\
        & \class{ `lz.
          \cas[{(\co_i\mapsto
            \proj[n]{i}(\proj{1}(\proj{1}z)))_{i=1}^n}]{
            \cas[{(\co_j\mapsto
              \proj[n]{j}(\proj{2}(\proj{1}z)))_{j=1}^n}]{
              (\proj{2}z)}}
        }\\
        & \class{ `lz.
          \cas[{(\co_j\mapsto
            \cas[{(\co_i\mapsto
              \proj[n]{i}(\proj{1}(\proj{1}z)))_{i=1}^n}]{
              \proj[n]{j}(\proj{2}(\proj{1}z))})_{j=1}^n}]{
            (\proj{2}z)}
        } \quad \text{\small (by \rul{CaseCase})}\\
      \end{array}
    \end{math}
  \item[{\drefFail:}]
    This diagram commutes if lhs~=~rhs, with lhs~=~$\proj{2}~;~\fail$,
    \\
    and rhs~=~$(\id[\pdt]`*\fail)~;~\case$.
    \\
    \begin{math}
      \begin{array}[b]{c!{=}l}
        \text{lhs} & \class[\pdt`*\uno→D]{
          `lz.(`lx.\cas[]{\co_1})~(\proj{2}z)} \\
        & \class[\pdt`*\uno→D]{`lz.\cas[]{\co_1}} \\
        \text{rhs} & \class[\pdt`*\uno→D]{`lz.t_{\case}~
          \pair{\proj{1}z,\,,(`lx.\cas[]{\co_1})~(\proj{2}z)}}\\
        & \class[\pdt`*\uno→D]{`lz.t_{\case}~
          \pair{\proj{1}z,\,,\cas[]{\co_1}}}\\
        & \class[\pdt`*\uno→D]{`lz.
          \cas[{(\co_i\mapsto
            \proj[n]{i}(\proj{1}z))_{i=1}^n}]{\cas[]{\co_1}}}\\
        & \class[\pdt`*\uno→D]{`lz.\cas[]{\co_1}} \qquad
        \text{\small (by \rul{CaseCase})}\\
      \end{array}
    \end{math}
    \qed
  \end{description}
\end{proof}

%%%%%%%%%%%%%%%%%%%%%%%%%%%%%%%%%%%%%%%%%%%%%%%%%%%%%%%%%%%%%%%%%%%%%%

\noindent
\textbf{Proposition \ref{prop:def-itp}.}
{\itshape
  In the model~\Ms, the interpretation of a term~$t$ in a context
  $`G=x_1;\cdots;x_k$ is 

  \hfill
  $\itp{t}=\class[D^k\to D]{`lx.\cpl{t}\subs[x_i]{\proj[k]{i}x}}$
  \hfill
  (with~$x$ fresh in~$t$).
}
\begin{proof}
  The proof proceeds by structural induction on~$t$.
  If~~$t=x_i$ or~$t=\co$, we just have to write the definition of~\itp{t}.
  If~$t=`lx_{k+1}.t_0$ or~$t=t_1t_2$, the equation is straightforward from definition of~\itp{t} and induction hypothesis.
  We detail the proof when~$t=\cas{u}$:\\
  $\itp{t}=\tuple{\itp{`q};\itp{u}};\case$,\ 
  with $\itp{`q}=\tuple{f_1,\dots,f_n}$ where $f_j=\itp{u_j}$ if $\co_j\mapsto u_j`:`q$, and $f_j=\term[D^k];\fail\ (=\class[D^k→D]{`lx.\cas[\ ]{\co[c_1]}})$ if~$\co_j`;\dom{`q}$.
  So
  \begin{displaymath}
    \itp{t}=\class[D^k→D]{`lx.t_{\case}~\pair{t_{`q}x,t_ux}}
  \end{displaymath}
  with $\case=\class[\pdt`*D→D]{t_{\case}}$,\ 
  $\itp{`q}=\class[D^k→\pdt]{t_{`q}}$,\ and
  $\itp{u}=\class[D^k→D]{t_u}$.
  By induction hypothesis, we can chose
  $t_u=`lx.\cpl{u}\subs[x_i]{\proj[k]{i}x}$, and
  $t_{`q}=`lx.\pair[n]{t_1x,\dots,t_nx}$ with 
  $t_j=`lx.\cpl{u_j}\subs[x_i]{\proj[k]{i}x}$ 
  if~$\co_j\mapsto u_j`:`q$, and $t_j=`lx.\cas[\ ]{\co[c_1]}$ if~$\co_j`;\dom{`q}$.\\
  Also
  \begin{math}
    \begin{array}[t]{c!{\eqlc}l}
      `lx.t_{\case}\,,\,\pair{t_{`q}x,t_ux} &
      `lx.t_{\case}\,,\,\pair{\, 
\pair[n]{t_1x,\dots,t_nx}\,,\,\cpl{u}\subs[x_i]{\proj[k]{i}x}\,}\\
      & `lx. \cas[{(\co_j\mapsto t_jx)_{j=1}^n}]
      {\cpl{u}\subs[x_i]{\proj[k]{i}x}}\\
      & `lx.\cpl{\cas{u}}\,\subs[x_i]{\proj[k]{i}}
    \end{array}
  \end{math}\\
  Indeed, $t_jx\eqlc\cpl{u_j}\subs[x_i]{\proj[k]{i}x}$ \  if~$\co_j\mapsto u_j`:`q$,\ \ and $t_j\eqlc\cas[\ ]{\co_1}$ \ if~$\co_j`;\dom{`q}$.\\
  Since~$D^k→D$ is compatible with~\eqlc,\  $\itp{t}=\class[D^k\to D]{`lx.\cpl{t}\subs[x_i]{\proj[k]{i}x}}$.
  \qed
\end{proof}

%%%%%%%%%%%%%%%%%%%%%%%%%%%%%%%%%%%%%%%%%%%%%%%%%%%%%%%%%%%%%%%%%%%%%%
%%%%%%%%%%%%%%%%%%%%%%%%%%%%%%%%%%%%%%%%%%%%%%%%%%%%%%%%%%%%%%%%%%%%%%

\subsection{Some rewriting properties}
\label{prf:cpl-rewrit}

\textbf{Lemme \ref{lem:cpl-red}.\ref{it:cpl-red-lcm}
  (\lcm reduction on completed terms).}\\
{\itshape
  Let~$t$ be a defined term.
  Then, for any term~$t'$,
  \begin{center}
    $\cpl{t}→_{\lcm}t'$\quad implies\quad $t'=\cpl{t_0}$\; for
    some~$t_0$ such that $t→t_0$.
  \end{center}
}
\begin{proof}
  By structural induction on~$t$.
  First notice that every \rul{CaseCons} redex present in~\cpl{t} corresponds to a \rul{CaseCons} redex in~$t$, as~$t$ is defined.
  Moreover, $\cas[\ ]{\co_1}$ is not reducible so every redex in a sub-term of~$\cpl{t}$ corresponds to a redex in a sub-term of~$t$
  Also if the reduction~$\cpl{t}→t'$ is performed in a (strict) sub-term of~$\cpl{t}$, we can immediately conclude with induction hypothesis.
  So it is sufficient to check the lemma for the five possible reductions in head position $\cpl{t}\rightarrowtriangle t'$, which is trivial.
  \qed
\end{proof}

%%%%%%%%%%%%%%%%%%%%%%%%%%%%%%%%%%%%%%%%%%%%%%%%%%%%%%%%%%%%%%%%%%%%%%

\noindent
\textbf{Lemme \ref{lem:cpl-red}.\ref{it:cpl-red-cc}
  (\rul{CaseCase} reduction on completed terms).}\\
{\itshape
  For any term~$t, t'$,
  \begin{center}
    $\cpl{t}→_{cc}t'$\quad implies\quad $t'→_{cc}^*\cpl{t_0}$\; for     some~$t_0$ such that $t→_{cc}t_0$
  \end{center}
}
\begin{proof}
  By by structural induction on~$t$.
  If the \rul{CaseCase} reduction occurs in a strict sub-term of $\cpl{t}$ then we conclude with induction hypothesis.
  Otherwise $t=\cas{\cas[`f]{u}}$, and $t'=\cas[\cpl{`q}`o\cpl{`f}]{\cpl{u}}$.
  Then we take~$t_0=\cas[`q`o`f]{u}$, since $\cpl{`q}`o\cpl{`f}→_{cc}^*\cpl{`q`o`f}$.
  Indeed, if $`f=\{\co_i\mapsto u_i/i`:I\}$ then 
  \begin{displaymath}
    \begin{array}[t]{ccl}
      \cpl{`q}`o\cpl{`f} &=&
      \{\co_i\mapsto \cas[\cpl{`q}]{\cpl{u_i}}/i`:I\} \cup
      \{\co_i\mapsto \cas[\cpl{`q}]{\cas[]{\co_1}}/i`;I\} \\
      \cpl{`q`o`f} &=&
      \{\co_i\mapsto \cas[\cpl{`q}]{\cpl{u_i}}/i`:I\} \cup
      \{\co_i\mapsto \cas[]{\co_1}/i`;I\}
    \end{array}
  \end{displaymath}
  Also $t'→_{cc}^*\cpl{t_0}$.
  \qed
\end{proof}

%%%%%%%%%%%%%%%%%%%%%%%%%%%%%%%%%%%%%%%%%%%%%%%%%%%%%%%%%%%%%%%%%%%%%%

\noindent
\textbf{Lemma \ref{lem:cnf-cpl}
  (Commutation case-completion/\rul{cc}-normal form).}\\
{\itshape
  For any term~$t$, $$\cnf{(\cpl{t})}=\cpl{\cnf{t}}~.$$
}
\begin{proof}
  By induction on the size of the maximal reduction  $\cpl{t}\tocc\cnf{(\cpl{t})}$.
  If $\cpl{t}=\cnf{(\cpl{t})}$, then $\cpl{t}$ is \rul{CaseCase}-normal, and so is~$t$ (Fact.\ref{fact:red-cpl}).
  Thus~$t=\cnf{t}$ and $\cpl{t}=\cpl{\cnf t}$.
  Otherwise let $\cpl{t}\tocc t'\tocc^* \cnf{(\cpl{t})}$.
  By Lem.~\ref{lem:cpl-red}.\ref{it:cpl-red-cc}, there is a term~$t_0$ such that $t'\tocc^*\cpl{t_0}$ and $t→_{cc}t_0$.
  Hence $\cpl{t} \tocc^+ \cpl{t_0} \tocc^* \cnf{(\cpl{t})} = \cnf{(\cpl{t_0})}$.
  By induction hypothesis, $\cnf{(\cpl{t_0})}=\cpl{\cnf{t_0}}$.
  Moreover $\cnf{t_0}=\cnf{t}$, so
  \begin{math}
    \cpl{(\cnf{t})} = \cpl{(\cnf{t_0})} = \cnf{(\cpl{t_0})} = \cnf{(\cpl{t})} 
  \end{math}.
  \qed
\end{proof}

%%%%%%%%%%%%%%%%%%%%%%%%%%%%%%%%%%%%%%%%%%%%%%%%%%%%%%%%%%%%%%%%%%%%%%

\noindent
\textbf{Lemma \ref{lem:lcm-cnf}.}
{\itshape
  For any terms~$t,t'$, if $t→_{\lcm}t'$ then there exists a term~$u$ such that
  \begin{displaymath}
    \cnf{t}→_{\lcm}^*u\tocc^*\cnf{t'}~.
  \end{displaymath}
}
\begin{proof}
  The proof proceeds by induction on~$\mes{t}$, the structural measure of~$t$ defined by
  \begin{displaymath}
    \begin{array}[t]{c!{=}c!{\qquad}c!{=}c!{\quad}c!{=}c}
      \mes x & 1 &  \mes{`lx.t} & \mes t + 1 &
      \mes{\cas{t}} & \mes t `* (\mes{`q} + 2)\\
      \mes \co & 1 & \mes{tu} & \mes t + \mes u &
      \mes{`q} & \sum_{\co`:\dom{`q}} \mes{`q_{\co}} \\
    \end{array}
  \end{displaymath}
  Notice that this measure decreases with the subterm relation but also with \rul{CaseCase} reduction ($\mes{\cas{\cas[`f]{u}}}>\mes{}\cas[`q`o`f]{u}$ for any~$`q,`f,t$).
  For any term~$s$ (or any case-binding~$`q$), $s'$ (\resp~$`q'$) represents a term (\resp\ a case-binding) such that $s→_{\lc}s'$ (\resp~$`q_{\co}→_{\lc}`q'_{\co}$ for some $\co`:\dom{`q}$, and $`q_{\co'}=`q'_{\co'}$ for $\co'\neq\co$)
  \begin{itemize}
  \item If $t$ is an application, either $t=t_1t_2$ and $t'=t'_1t_2$ (or $t'=t_1t'_2$) and we conclude with induction hypotheses, or $t=(`lx.t_1)t_2$ and $t'=t_1\subs{t_2}$.
    In that case, $\cnf{t}=(`lx.\cnf{t_1})\cnf{t_2} \to_{\lcm} (\cnf{t_1})\subs{\cnf{t_2}} \tocc^* \cnf{(\cnf{t_1})\subs{\cnf{t_2}}}$.
    Moreover, $\cnf{(\cnf{t_1})\subs{\cnf{t_2}}} = \cnf{(t_1\subs{t_2})}$.
    Thus $\cnf{t} \to_{\lcm} (\cnf{t_1})\subs{\cnf{t_2}} \tocc^* \cnf{t'}$.
  \item If $t$ is an abstraction, either $t=`lx.t_0$ and $t'=`lx.t'_0$ and we conclude with induction hypothesis, or $t=`lx.t'x$ with~$x`;\fv{t'}$.
    In that case, $\cnf{t}=`lx.\cnf{t'}x →_{\lcm} \cnf{t'}$.
  \item If $t=\cas{x}$, then $t'=\cas[`q']{x}$ and we conclude with induction hypothesis.
  \item If $t=\cas{\co}$, then either $t'=\cas[`q']{\co}$ and we conclude with induction hypothesis, or $t'=`q_{\co}$ and $\cnf{t}=\cas[\cnf{`q}]{\co}→_{\lcm}\cnf{`q_{\co}}$.
  \item If $t=\cas{t_1t_2}$, then either $t'=\cas[`q']{t_1t_2}$ and we conclude with induction hypothesis, 
    or $t'=\cas{t_0}$ with $t_1t_2→_{\lcm}t_0$ or $t'=(\cas{t_1})t_2$.

    In the second case, by induction hypothesis there is some~$u_0$ such that \linebreak[4] $\cnf{t_1t_2}→_{\lcm}^*u_0\tocc^*\cnf{t_0}$.
    Hence
    \begin{displaymath}
      \cnf{t}=\cas[\cnf{`q}]{\cnf{t_1t_2}} →_{\lcm}^* \cas[\cnf{`q}]{u_0}   \tocc^* \cas[\cnf{`q}]{\cnf{t_0}} \tocc^*   \cnf{\cas[\cnf{`q}]{\cnf{t_0}}}~.
    \end{displaymath}
    Moreover, every sub-term of~\cnf{t'} is in \rul{CaseCase} normal form, so \linebreak[4] $\cnf{t'}=\cnf{\cas[\cnf{`q}]{\cnf{t_0}}}$.
    Thus
    \begin{math}
      \cnf{t} →_{\lcm}^* \cas[\cnf{`q}]{u_0} \tocc^* \cnf{t'}
    \end{math}.
    \\
    In the last case, $\cnf{t}=\cas[\cnf{`q}]{(\cnf{t_1}\cnf{t_2})}$, so

    \hfil
    \begin{math}
      \cnf{t} →_{\lcm}
      (\cas[\cnf{`q}]{\cnf{t_1}})\cnf{t_2} \tocc^*
      \cnf{(\cas[\cnf{`q}]{\cnf{t_1}})}\cnf{t_2} =   \cnf{\cas{t_1}}\cnf{t_2}
    \end{math}.

  \item If $t=\cas{`lx.t_0}$, idem as previous case.
  \item If $t=\cas{\cas[`f]{t_0}}$, then either $t'=\cas{\cas[`f']{t_0}}$, or $t'=\cas{\cas[`f]{t'_0}}$, or $t'=\cas[`q']{\cas[`f]{t_0}}$.

    In the first case, write $t_1=\cas[`q`o`f]{t_0}$ and $t'_1=\cas[`q`o`f']{t_0}$.
    Remark that $\mes{t_1}<\mes{t}$ (since the structural measure decreases by \rul{CaseCase}-reduction), \linebreak[4] and that $t_1→_{\lc}t'_1$.
    By induction hypothesis, there is some~$u$ such that \linebreak[4]
    \begin{math}
      \cnf{t_1}→_{\lcm}^*u\tocc^*\cnf{t'_1}
    \end{math}.
    Since $\cnf{t}=\cnf{t_1}$ and $\cnf{t'}=\cnf{t'_1}$ we are done.

    In the second case, same method but with $t'_1=\cas[`q`o`f]{t'_0}$.

    In the last case, write $t=\cas{\cas[`f_1]{\cdots\cas[`f_k]{u_0}}}$, where $u_0$ is not a case construct (thus $k\geq 1$).
    Then $\cnf{t}=\cas[\cnf{(`q`o`j)}]{\cnf{u_0}}$, with $`j=`f_1`o(\cdots`o`f_k)$, and $\cnf{t'}=\cas[\cnf{(`q'`o`j)}]{\cnf{u_0}}$ (since $((`q`o`f_1)`o\cdots)`o`f_k \tocc^* `q`o`j$).

    Let us explicit~\cnf{t}\ and~\cnf{t'}:
    \begin{math}
      \begin{array}[t]{c!{=}l}
        \cnf{t} & \cas[
        \co\mapsto\cnf{\cas{`j_{\co}}}\,/\,\co`:\dom{`j}
        ]{\cnf{u_0}}\\
        \cnf{t'} & \cas[
        \co\mapsto\cnf{\cas[`q']{`j_{\co}}}\,/\,\co`:\dom{`j}
        ]{\cnf{u_0}}\\
      \end{array}
    \end{math}

    Remark that $\mes{\cas{`j_{\co}}}\leq\mes{t}$ (the structural measure decreases by \rul{CaseCase}-reduction, and preserves the order of sub-term relation), and that \linebreak[4] $\cas{`j_{\co}}→_{\lcm}\cas[`q']{`j_{\co}}$.
    Hence , by induction hypothesis, for each~$\co`:\dom{`j}$ there is a term~$u_{\co}$ such that
    $\cnf{\cas{`j_{\co}}}→_{\lcm}^*u_{\co}\tocc^*\cnf{\cas[`q']{`j_{\co}}}$.
    Thus

    \hfill
    \begin{math}
      \cnf{t}→_{\lcm}^*u\tocc^*\cnf{t'}
      \qquad\text{for}\qquad
      u=\cas[\co\mapsto u_{\co}\,/\,\co`:\dom{`j}]{\cnf{u_0}}~.
    \end{math}
    \qed
  \end{itemize}
\end{proof}

%%% Local Variables: 
%%% mode: latex
%%% TeX-master: "catLC"
%%% End: 

\end{document}